\documentclass[12pt]{article}

\usepackage{epsfig}

\usepackage{amssymb}
\usepackage{amsfonts}

\usepackage{color}
 
\def\be{\begin{equation}}
\def\ee{\end{equation}}
\def\ba{\begin{array}{c}}
\def\ea{\end{array}}
\def\p{\partial}
\def\ben{$$}
\def\een{$$}

\newcommand{\bea}{\begin{eqnarray}}
\newcommand{\eea}{\end{eqnarray}}

\newtheorem{thm}{Theorem}

\newtheorem{lemma}[thm]{Lemma}

\newtheorem{rem}[thm]{Remark}
\newenvironment{proof}{\noindent
 {\bf Proof.}}{\hfill$\square$\vspace{3mm}\endtrivlist}

\begin{document}

\begin{center}

.\vspace{.1cm}

 \begin{center}{\Large \bf

Calogero{-like} model without rearrangement symmetry



  }\end{center}

\vspace{0.28cm}

  {\bf Miloslav Znojil}$^{a,b}$

\end{center}

\vspace{6mm}

 $^{a}${The Czech Academy of Sciences,
 Nuclear Physics Institute, 
 Hlavn\'{\i} 130,
250 68 \v{R}e\v{z}, Czech Republic, {{e-mail:
znojil@ujf.cas.cz}}}


 $^{b}${Department of Physics, Faculty of
Science, University of Hradec Kr\'{a}lov\'{e}, Rokitansk\'{e}ho 62,
50003 Hradec Kr\'{a}lov\'{e},
 Czech Republic}

\vspace{5mm}


\section*{Abstract}

Reinterpretation of mathematics behind the exactly solvable
Calogero's $A-$particle quantum model is used to propose its
generalization. Firstly, {it is argued} that the strongly singular
nature of the Calogero's particle-particle interactions makes {the}
original permutation-invariant Hamiltonian tractable as a direct sum
$H=\bigoplus\,H_a$ of isospectral components which are mutually
independent. Secondly, {after the elimination of} the center-of-mass
motion the system is reconsidered as living in the reduced Euclidean
space $\mathbb{R}^{A-1}$ of relative coordinates {and decaying} into
a union of subsets $W_a$ called Weyl chambers. {} The mutual
independence of the related reduced forms of operators $H_a$ enables
us to choose them non-isospectral. This {breaks} the symmetry {and
unfolds} the spectral degeneracy of $H$. A new, multi-parametric
generalization of the conventional $A-$body Calogero model is
obtained. Its detailed description is provided up to $A=4$.

\section*{Keywords}

Calogero's $A_N$ model; asymmetric two-particle barriers;
exact solvability requirement; coloring of the
Weyl chambers;

\newpage

\section{Introduction \label{uvod} }

The well known \cite{perelomov,perelomovb,turbiner} quantum Hamiltonian
 \be
 H^{(A)}(\omega,C) = - \sum_{i=1}^{A}\
\frac{\p^2}{\p{x_i}^2}
 + \sum_{ i<j=2}^{A}\left [
\frac{1}{8}\,\omega^2   \,(x_i-x_j)^2
 +\frac{2C}{
 (x_i-x_j)^{2}}\right ]\,,\ \ \ \ x_k \in \mathbb{R}\,,
 \ \ \ \ C>-\frac{1}{4}
 \label{CaHa}
  \ee
invented by Calogero \cite{Calogero,Calogerox,Calogerob}
and written here in
units $\hbar =2m=1$ plays an important methodical role in
nuclear, atomic and molecular physics
and in quantum chemistry \cite{vinet}.
It
is a truly remarkable one-dimensional $A-$particle-chain
{\it alias\,} linear-molecule
model which
combines the non-numerical solvability \cite{turbinerd}
with a nontrivial
and multi-branched
phenomenological relevance \cite{suthe,polych}.
Its symmetry with respect to the
permutations of coordinates $x_k \in \mathbb{R}$ is
accompanied by a fairly realistic shape of its two-body interaction
potentials mimicking not only the expected asymptotic attraction but
also a frequently encountered repulsion at short distances
\cite{scholarpedia}.

Equally strongly the impact of the model can be felt
in mathematics.
The wealth of related innovations ranges from the
upgrades of the applications of Lie algebras
\cite{ryu,ryub,turbinerb}
and of
orthogonal polynomials
\cite{dunkl,vandiejen}
up to the
amendments of paradigms known as Wigner-Dunkl
quantum mechanics \cite{WD,WDb},
quasi-Hermitian quantum mechanics \cite{Geyer,tater,translucent,taterd,fringb,fringbb}
or ${\cal PT}-$symmetric
quantum mechanics
\cite{book,ptho,taterb,taterc,fring,fringc,correa}.

In our paper we intend to propose a multiparametric but {still}
exactly solvable generalization of model (\ref{CaHa}). For {}
introduction {} it is sufficient to consider just the most
elementary two-particle special case
 \be
 H^{(2)}(\omega,C) = -
\frac{\p^2}{\p{x_1}^2} - \frac{\p^2}{\p{x_2}^2}
 +
\frac{1}{8}\,\omega^2   \,(x_1-x_2)^2
 +\frac{2C}{
 (x_1-x_2)^{2}}\,,
 \ \ \ \ (x_1, x_2) \in \mathbb{R}^2\,.
 \label{CaHa2}
  \ee
{The} well known additional merit of the model emerges after one
defines the two new  ``relative-motion'' {{\it alias\,} Jacobi}
coordinates \cite{turbiner}
 \be
  R=\frac{1}{\sqrt{2}}(x_1+x_2)\,,\ \ \ \ X=\frac{1}{\sqrt{2}}(x_1-x_2)
  \label{tri}
 \ee
and after one manages to separate the center-of-mass
motion \cite{haka0b,haka0}. This
leads to the reduction and replacement of the initial partial
differential Schr\"{o}dinger equation by another, ordinary
differential bound-state problem in $L^2(\mathbb{R})$,
 \be
 \left [ -\frac{d^2}{d X^2} +
\frac{1}{4}\,\omega^2X^2+\frac{C}{X^{2}} \right ]
\psi_n(X) = E_n\,\psi_n(X)\,,\ \ \ \ \ \ X \in \mathbb{R}\,,\ \ \ \
\ \ \ n=0,1,\ldots\,.
 \label{SE2k}
 \ee
The positive $X\in \mathbb{R}^+$ corresponds to the ordering
$x_1>x_2\,$ (with the first particle lying, on
the real line ${\mathbb{R}}$, to the right
from the second one) while the negative choice of $X\in \mathbb{R}^-$
represents our pair of particles as
positioned inside a complementary half-plane of
$(x_1,x_2)\in \mathbb{R}^2$ where $x_1<x_2$.

The two half-lines $\mathbb{R}^{\pm} \subset \mathbb{R}$ may be
called Weyl chambers. In the literature, such a name is used at any
number of particles $A$. After the standard elimination of the
center of mass \cite{haka1,haka2} the $A-$particle-coordinate space
${\mathbb{R}}^A$ becomes reduced to its subspace $\mathbb{R}^{A-1}$.
This is performed in full analogy with the $A=2$ change of
coordinates (\ref{tri}) so that at any higher $A>2$, the space
$\mathbb{R}^{A-1}$ remains parametrized by the relative {} Jacobi
coordinates \cite{turbiner,haka3,haka4}.

In terms of the latter coordinates the singular
repulsion enters the game and the reduced space
becomes split,
in a way generalizing the $A=2$ case,
into a union of $(A-1)-$dimensional Weyl chambers
$W_a$, i.e.,
$\mathbb{R}^{A-1}=\bigcup\,W_a$.
Unfortunately, a return to $A=2$ and to the reduced
Schr\"{o}dinger equation (\ref{SE2k}) reveals the existence of a
subtle mathematical problem, well known to all of the authors of
textbooks. In a way emphasized, e.g., by Landau and Lifshitz
\cite{Landau}, the short-range force $\sim X^{-2}$ is strongly
singular in the origin so that the Calogero's system
living on the real line $\mathbb{R}$ becomes tractable as composed
of two {\em completely dynamically independent\,} quantum systems
living on the respective permutation-characterized Weyl-chamber
half-lines $\mathbb{R}^{-}=W_{(\{12\})}$ or $\mathbb{R}^{+}=W_{(\{21\})}$.
All of the meaningful phenomenological predictions
(i.e., say, of the structure of the
spectrum or of the wave functions)
are then encoded in any one of its Weyl chambers.
In this sense, the reference to the ``global'' Hamiltonians (\ref{CaHa2})
or (\ref{CaHa})
can be considered (perhaps, unintentionally) misleading.

Such an observation has several consequences, some of which will be
explained and described in what follows. The presentation of our
results will be preceded by a brief review of some of the basic
properties of the most elementary conventional Calogero model in
section \ref{revisited}. As a core and guide to our project, the
``asymmetrization'' {preserving} the solvability will be proposed
there at $A=2$. In {} subsequent section \ref{3revisited} the
details of analogous asymmetrization will be described at~$A=3$.
After transition to arbitrary $A$,  our innovated Calogero-like
model will finally be analyzed and discussed in the last two
sections \ref{4revisited} and \ref{summary} and in Appendix.

\section{Spectral degeneracy and its unfolding at $A=2$\label{revisited}}

The ${\cal O}(x^{-2})$ singularity in Eq.~(\ref{SE2k}) resembles the
centrifugal term in the radial Schr\"{o}dinger equation of a
centrally symmetric harmonic oscillator in a specific $\ell-$th
partial wave. Formally, we may reparametrize $C=\ell(\ell+1)$ and
reconstruct the real (though, in general, non-integer and/or
non-positive) angular-momentum-like parameter $\ell$ from a given
value of{} coupling constant $C$,
 \be
\ell =\ell(C) = -\frac{1}{2} + \sqrt{ \frac{1}{4}+{C} } \ .
\label{fu}
 \ee
A decisive difference from  the radial Schr\"{o}dinger equation is
that our present Eq.~(\ref{SE2k}) has to live, by definition, on the
whole real line of $X \in \mathbb{R}$. In the Calogero's operator
$H^{(2)}$ of Eq.~(\ref{CaHa2}) the two particle coordinates $x_1$
and $x_2$ are independent variables so that one has to represent
{the} motion in both the left Weyl chamber $W_{(\{12\})}$ of $X\in
(-\infty,0)$ and the right Weyl chamber $W_{(\{21\})}$ of $X\in
(0,\infty)$.

\subsection{Singularity in the origin}

Due to the singularity at $X=0$ our reduced Hamiltonian is merely
essentially self-adjoint~\cite{Simon}. In the origin some
additional boundary conditions have to
be imposed in order to make the problem mathematically well defined.
An explicit specification of physics behind these boundary
conditions is necessary. Only then one can speak about a
consistent quantum theory and about a unique operator associated
with the differential expressions of Eq.~(\ref{CaHa})
or,
at
$A=2$, of Eqs.~(\ref{CaHa2}) and~(\ref{SE2k}).

In this light we will accept,
at
$A=2$, the most common
convention by which one requires
 \be
  \psi_n(X) \sim X^{\ell+1}\,,\ \ \ \ \ X \sim 0\,,\ \ \ \ \
 \ell>-1/2
  \,.
  \label{constra}
 \ee
This is equivalent to the suppression of the dominant component $
\sim X^{-\ell}$ of the wave function near the origin (see a few
related comments in \cite{PRAZnojil}). {The} Calogero's reduced
symmetric-interaction Eq.~(\ref{SE2k}) can be {then} perceived as
living on the full real line. In spite of the presence of the
barrier, the bound states are made well defined by constraint
(\ref{constra}).

In a more consequent and more physics-oriented conceptual setting
the interpretation of the role of the barriers is less clear. The
tunneling between the two neighboring half-lines {\it alias\,} Weyl
chambers is fully suppressed. {One} could speak about a direct sum
of the two independent quantum systems. Moreover, once we recall the
even-parity symmetry of the potential in (\ref{SE2k}), we have a
full freedom of redirecting our attention from the model defined on
the whole real line back to just one of the independent sub-models.
Both of these sub-models are isospectral so that the spectrum of the
complete system is doubly degenerate.

In many applications, people decide to ignore the impenetrability of
the barrier (i.e., the absence of any meaningful contact between the
subsystems) {preferring} the reference to the full-line Hamiltonian
(\ref{CaHa}). Such a slightly manipulative decision is presented as
well motivated by the possibility of working, mostly in the context
of statistical physics, with the two alternative versions of the
wave functions which are constructed as spatially symmetrized or
antisymmetrized,
 \be
 \psi_n^{(symm/antisymm)}(X)
 =\frac{1}{\sqrt{2}}[\psi_n(X)\pm \psi_n(-X)]
 \,,\ \ \ \ n=0,1,\ldots
 \,.
 \label{ugos}
 \ee
In spite of a complete absence of tunneling, the respective
wave-function superpositions (\ref{ugos}) are declared tractable
as mimicking a system of two indistinguishable
``bosons'' (due to the Pauli-principle-simulating symmetry
$\psi(-X)=\psi(X)$) or  ``fermions'' (with the spatial antisymmetry
sampling the fermionic statistics).

Due to the degeneracy in combination
with the absence of tunneling
one could equally well decide to consider some more
sophisticated full-line requirements with, say,
wave functions such that $\psi(-X)=q\,\psi(X)$, $X>0$
using an arbitrary complex $q \in {\mathbb C}$. Nevertheless,  in
spite of being acceptable mathematically, such a generalization
is not used. In the overall pragmatic context of the
applied quantum mechanics, even the two most elementary specifications of
$q = \pm 1$ are believed to make the most elementary one-dimensional
two-particle version (\ref{SE2k}) of the conventional Calogero model
as well as all of its $A>2$ descendants sufficiently
appealing, intuitive and useful.

In our present paper we intend to advocate a different philosophy.

\subsection{Asymmetric barrier\label{sollasi}}

During the recent developments of quantum theory the
Pauli-principle-rooted paradigm seems to be shattered. The first,
purely mathematical reason is that once we have zero tunneling
through the Calogero's barrier, there is no reason for keeping this
barrier formally left-right symmetric. In Eq.~(\ref{SE2k}) with
$C=\ell(\ell+1)$, therefore, we can feel free to ``asymmetrize''
the singular term,
 \be
 \frac{\ell(\ell+1)}{X^2} \ \to \
 \left \{
 \ba
 \ell_{(left)}(\ell_{(left)}+1)
  /X^2 \,,\ \ \ \ \ X < 0\,, \\
 \ell_{(right)}(\ell_{(right)}+1)
 /X^2 \,,\ \ \ \ X > 0\,,
 \ea
 \right .\,\ \ \ \  \ell_{(left)}\neq \ell_{(right)}
 \,.
  \label{ainq}
 \ee
Precisely such an idea of asymmetrization of the
impenetrable barrier at $A=2$ served as an inspiration of our
present paper.

The idea
appeared supported, independently, by some very recent developments
in applications. The necessity of working with an
asymmetric singular barriers emerged, for example, in the context of
relativistic quantum mechanics \cite{Ishkhab}. The presence of
interactions containing an asymmetric impenetrable barrier
has been found productive, especially for certain systems
described by one-dimensional Dirac equations (cf. \cite{Ishk}).
In
another paper (cf. \cite{Ishkha}) the same authors introduced and solved
Dirac equation in which the impenetrable barrier has been kept
left-right symmetric. The routine reduction of the equation
led to an equivalent Schr\"{o}dinger-type differential equation
defined along the whole real line in which the effective ${\cal
O}(x^{-2})$ barrier re-appeared in a manifestly asymmetric
form of Eq.~(\ref{ainq}) (see also a few more comments
on this topic in Appendix).

Let us now accept the idea
and let us replace,
in our reduced Calogero-Schr\"{o}dinger
equation~(\ref{SE2k}),
the conventional
symmetric ${\cal O}(x^{-2})$ barrier
by
its two-parametric asymmetric generalization (\ref{ainq}).
Such a model (i.e., two decoupled
radial harmonic oscillators with different angular momenta)
is exactly solvable \cite{Fluegge}.
Its spectrum can be written
in closed form,
 \be
 E_{}=4k+2\,\ell+3\,,
 \ \ \  k=0,1,2,\ldots\,,
 \ \ \ \ell=\ell_{{left}}>-1/2\,\ \ \ {\rm or}
 \ \ \ \ell=\ell_{{right}}>-1/2
 \,.
 \label{jehosp}
 \ee
{The} degeneracy as encountered when $ \ell_{{left}}=\ell_{{right}}
$ becomes, up to accidental confluences, removed when $
\ell_{{left}}\neq \ell_{{right}}$.

Due to
the simplicity of the model it becomes entirely
straightforward to deduce the basic consequences.
In a methodically motivated analysis let us
consider, nevertheless, just
a
special case of Eq. (\ref{ainq}) in which the asymmetry is
kept one-parametric and maximal,
 \be
 \ell_{left}=\ell_{left}(y)=y\,,\ \ \ \ \ell_{right}=\ell_{right}(y)=-y\,.
 \label{las}
 \ee
After we
recall Eq.~(\ref{SE2k}) and
out-scale the inessential spring constant
$\omega \to 2$ we arrive at a drastically
asymmetric one-parametric toy-model potential
 \be
  V_{}(X,y) = \left \{
 \begin{array}{ll}
 X^2+y(y+1)/X^{2} \,,\ \ \  \ \ &X < 0\,,\\
 X^2+(y-1)y/X^{2}\,,\ \ \  \ \ &X > 0\,,
 \ea
 \right .\,
 \ \ \ y = {\rm real}\,
 \label{kvIK}
 \ee
the shape of which is sampled here in Figure \ref{picee3ww}.

%
%
\begin{figure}[h]                     
\begin{center}                         
\epsfig{file=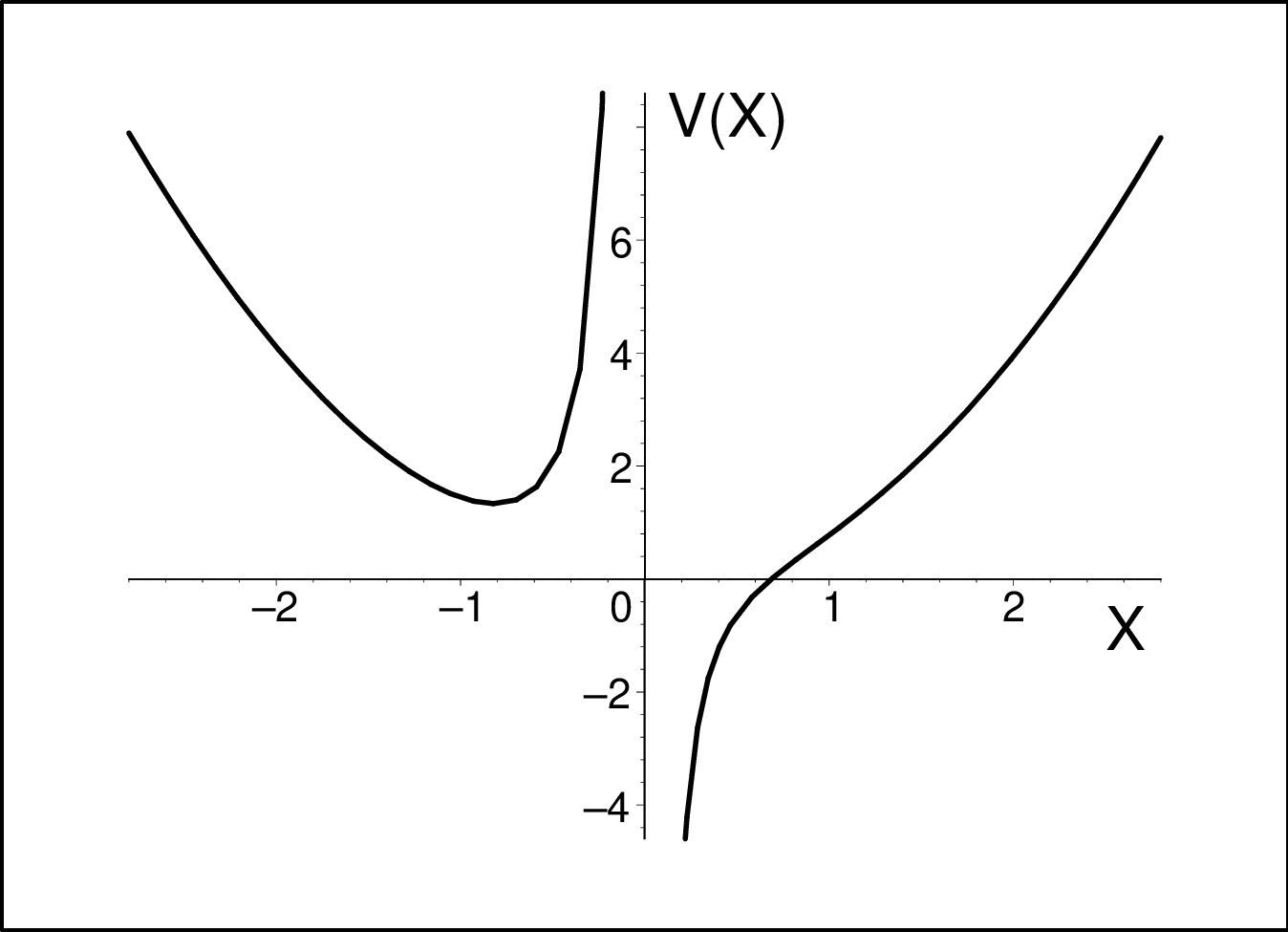,angle=270,width=0.4\textwidth}
\end{center}                         
\vspace{-2mm}\caption{Potential (\ref{kvIK}) at $y=1/3$.
 \label{picee3ww}}
\end{figure}

%
%
%
%
\begin{figure}[h]                     
\begin{center}                         
\epsfig{file=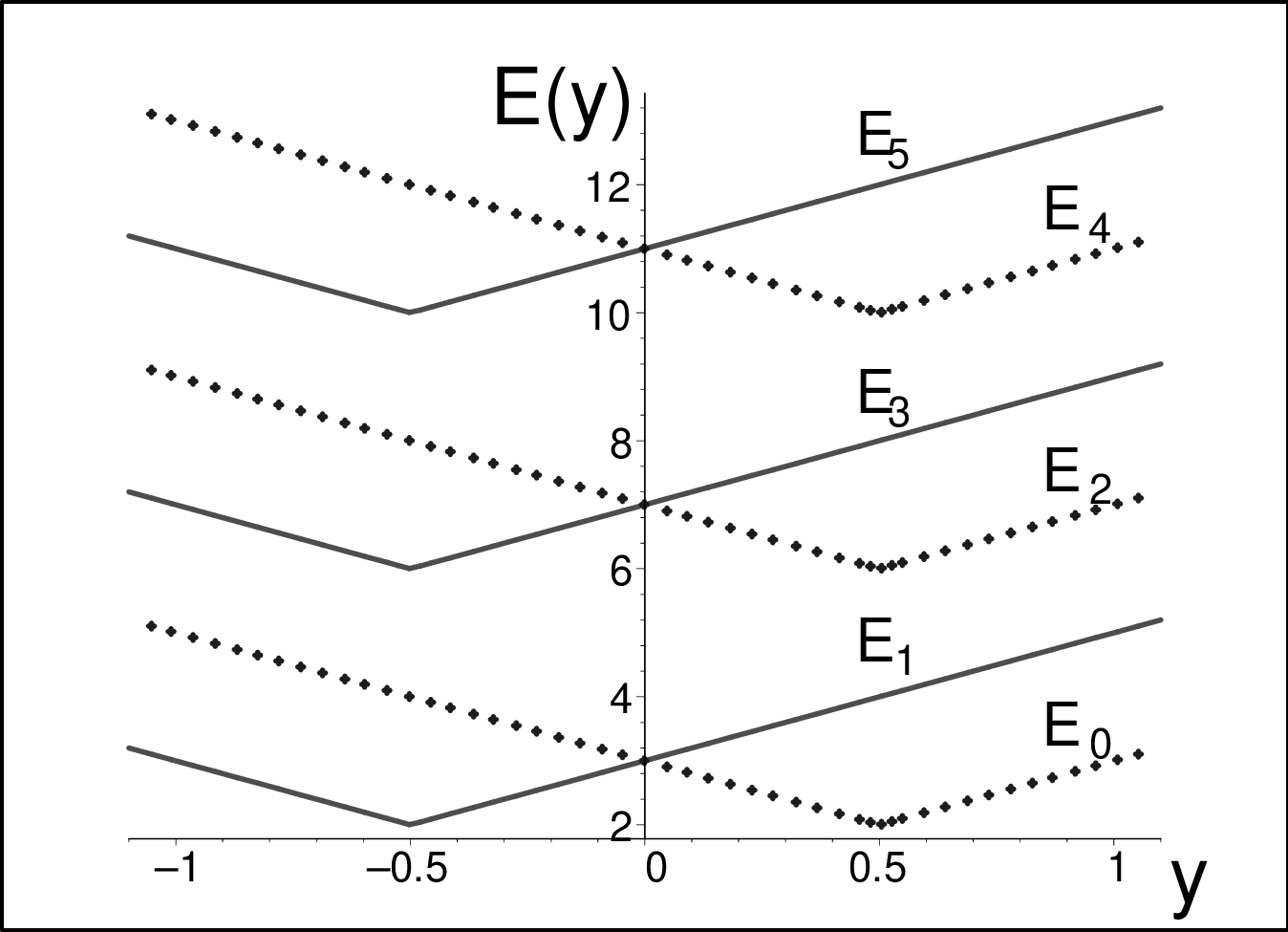,angle=270,width=0.5\textwidth}
\end{center}                         
\vspace{-2mm}\caption{Bound-state energies in
(\ref{kvIK}). Full/dotted lines mark the left-/right-well sub-spectra.
 \label{icee3ww}}
\end{figure}
%

The explicit non-numerical form of the
spectrum
immediately follows from Eq.~(\ref{jehosp})
and its $y-$dependence is
displayed in Figure \ref{icee3ww}.
The only instant of
degeneracy occurs at $y = 0$.
The spectrum even becomes
equidistant at $|y|>1/2$.
One only has to add that
the price to be paid for the
unfolding of the degeneracy is in fact not too small:

\begin{lemma}
\label{predlemma}
Schr\"{o}dinger Eq.~(\ref{SE2k}) with
one-parametric
asymmetric potential (\ref{kvIK})
yields, at $y \neq 0$,
the non-degenerate
bound-state spectrum
(\ref{jehosp}),
but all of its ``left''
eigenfunctions
$\psi_{n}(X)$
such that
$X \in W_{(\{12\})}=(-\infty,0)\,$
vanish
along the ``right'' Weyl-chamber
half-line  $W_{(\{21\})}=(0,\infty)\,$ and {\rm vice versa}.
\end{lemma}

Analogous consequences can be drawn after
the more general
two-parametric asymmetrization of the barrier (cf. (\ref{ainq})),
or after transition to a larger number of particles.
With the details to be explained below,
let us only mention here
that once we, at any $A$,
asymmetrize the boundary between
some two Weyl chambers $W_a$ and $W_b$
by
the choice of the respective couplings $C_a \neq C_b$,
we reveal
that
in a way sampled by Lemma \ref{predlemma} and
Figure \ref{icee3ww}, the
composition of the
two sub-spectra may remain non-degenerate.
Up to the accidental degeneracies, therefore,
the $W_a-$supported  wave function $\psi_a$
will necessarily vanish in $W_b$ and
{\it vice versa}.

In the limit of $C_a = C_b$ the spectrum becomes
degenerate
so that we may
return to the wave-function superpositions as
sampled by Eq.~(\ref{ugos}).
Due to the
absence of tunneling,
there is still no reason for
leaving the
Weyl-chamber-dependent direct-sum
interpretation
of Hamiltonians (\ref{CaHa}).

We may now return to the two-body case
in which the following statement is immediate.

\begin{thm}
\label{preteor}
After the asymmetrization (\ref{ainq}) of the potential in Eq.~(\ref{SE2k}),
we merely have to replace
the conventional
two-parametric Calogero Hamiltonian of Eq.~(\ref{CaHa2})
by its three-parametric
Weyl-chamber-dependent asymmetric-barrier generalization
 \be
 H^{(2)}(\omega,C_{(left)},C_{(right)}) =
 \left \{
 \ba
 -\frac{\p^2}{\p{x_1}^2} - \frac{\p^2}{\p{x_2}^2}
 +
 \frac{1}{8}\,\omega^2   \,(x_1-x_2)^2
 +\frac{2C_{(left)}}{
 (x_1-x_2)^{2}}\,,\ \ \ \ \ x_1<x_2 \,, \\
 -\frac{\p^2}{\p{x_1}^2} - \frac{\p^2}{\p{x_2}^2}
 +
 \frac{1}{8}\,\omega^2   \,(x_1-x_2)^2
 +\frac{2C_{(right)}}{
 (x_1-x_2)^{2}}\,,\ \ \ \ \ x_1>x_2 \,. \\
  \ea
 \right .\,
 \,
  \label{ainq2}
 \ee
\end{thm}
Strictly speaking the couplings
$C_{(left)}=\ell_{(left)}(\ell_{(left)}+1)$ and $
C_{(right)}=\ell_{(right)}(\ell_{(right)}+1) $
may though need not be different.
In both of these scenarios
the main message delivered by Theorem~\ref{preteor}
is that
if we wish to
understand and to
break the symmetry of the
system (represented, in this section, by the $A=2$
Hamiltonian
$H^{(2)}(\omega,C)$ of Eq.~(\ref{SE2k}))
and if we wish to introduce its consistently asymmetrized
generalization
(sampled, say, via Figure \ref{picee3ww} above),
we just have to take the two {\em different\,}
versions of the Hamiltonian
and we have to {\em restrict them} to the
respective single Weyl chamber.
In a way prescribed by Eq.~(\ref{ainq2}),
and in a way which will be generalized below:
A
similar conclusion will apply to the
asymmetrized
Calogero-like systems of more particles.

As long as the related analysis would be technically less
transparent, it makes sense to proceed {more slowly} and to move,
first, to the very next special case with three particles.

\section{Spectral degeneracy and its unfolding at $A=3$\label{3revisited}}

For some
purposes
(and, in particular, for methodical purposes)
the original symmetric Calogero model (\ref{CaHa})
as  well as its present
asymmetric Calogero-like generalization
are too elementary at $A=2$.
In both of these regimes,
fortunately, the
pedagogical merits of these models
are only partially lost after transition to
$A=3$ \cite{tater}.


%
\begin{figure}[h]                     
\begin{center}                         
\epsfig{file=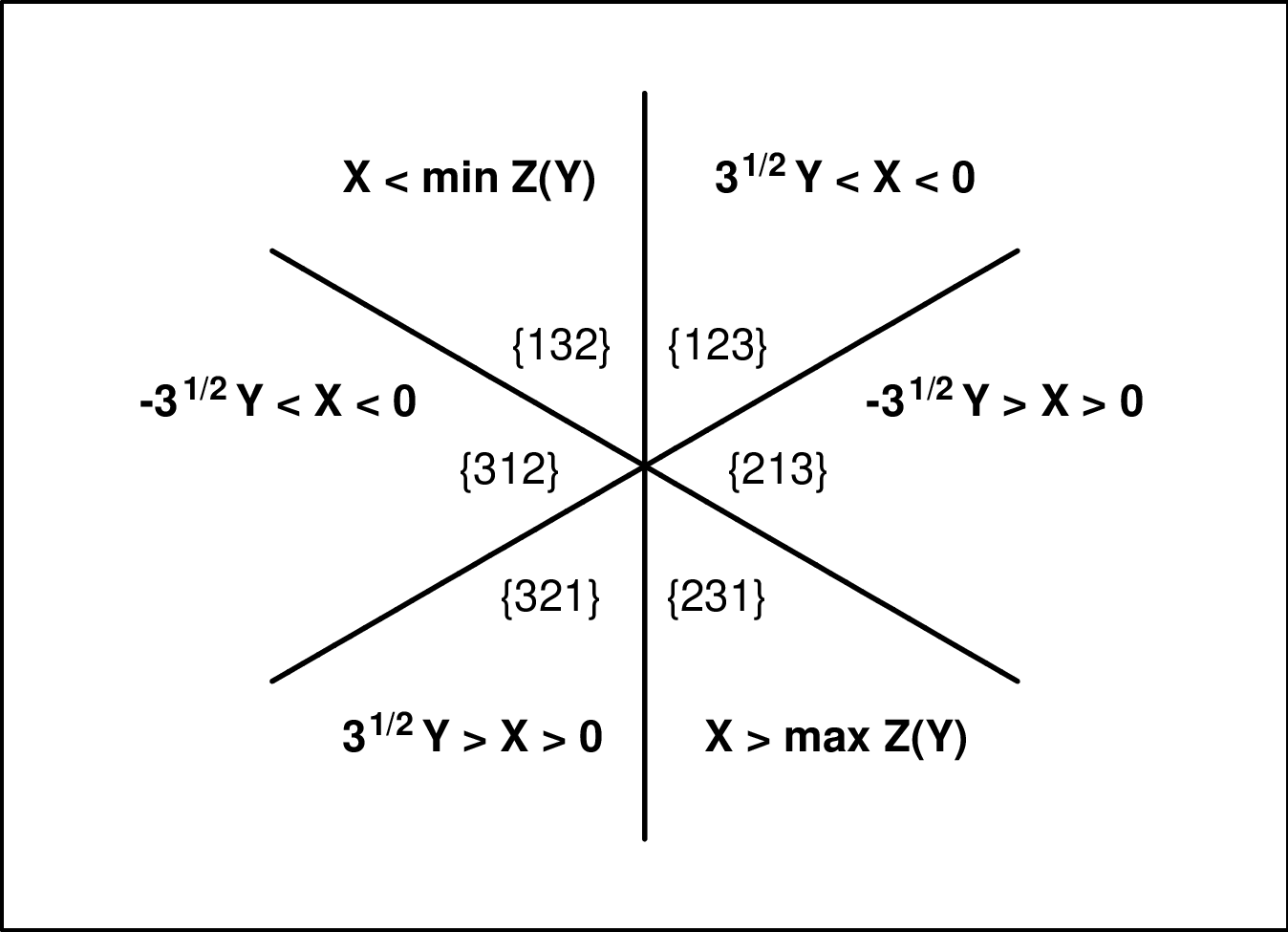,angle=270,width=0.4\textwidth}
\end{center}                         
\vspace{-2mm}\caption{Six wedge-shaped
Weyl chambers $W_{\{ijk\}}$
in $X-Y$ plane at $A=3$.
 \label{3wwu}}
\end{figure}

\subsection{Weyl chambers\label{setri}}

 \noindent
In the Calogero's symmetric model
(\ref{CaHa}) with $A=3$
the change of variables
 \be
 R = \frac{1}{\sqrt{3}}\,(x_1+x_2+ x_3)\,, \ \ \ \ \
  X = \frac{1}{\sqrt{2}}\,(x_1-x_2)\,, \ \ \ \ \
 Y = \frac{1}{\sqrt{6}}\,(x_1+x_2-2\,x_3)
 \label{ular}
  \ee
enables us to eliminate the center-of-mass motion and to
arrive at
the
reduced but still
partial differential Schr\"{o}dinger equation in the $X-Y$ plane $\mathbb{R}^2$,
 \be
\left [ -\frac{\p^2}{\p X^2} -\frac{\p^2}{\p Y^2}
+\frac{3}{8}\,\omega^2(X^2+Y^2) +\frac{C}{X^{2}}
+\frac{C}{(\sqrt{3}Y-X)^{2}}
+\frac{C}{(\sqrt{3}Y+X)^{2}}-E \right ] \,\Phi(X,Y)=0\,.
 \label{gular}
 \ee
At any non-vanishing coupling constant $C > -1/4$
the three impenetrable barriers
divide
the
plane
into as many as six
Weyl-chamber
sectors.
Their shapes are
displayed in
Figure \ref{3wwu}.
We define them there not only
by the inequalities
imposed upon $X$ and $Y$
(in which the symbol $Z(Y)$ stands for the
set $Z(Y)=\{+ \sqrt{3}\,Y,- \sqrt{3}\,Y\}$ of two elements) but also,
equivalently, by the subscripts in  $W_{\{ijk\}}$
which
refer to the respective particle orderings $x_i<x_j<x_k\,$.


%
\begin{figure}[h]                     
\begin{center}                         
\epsfig{file=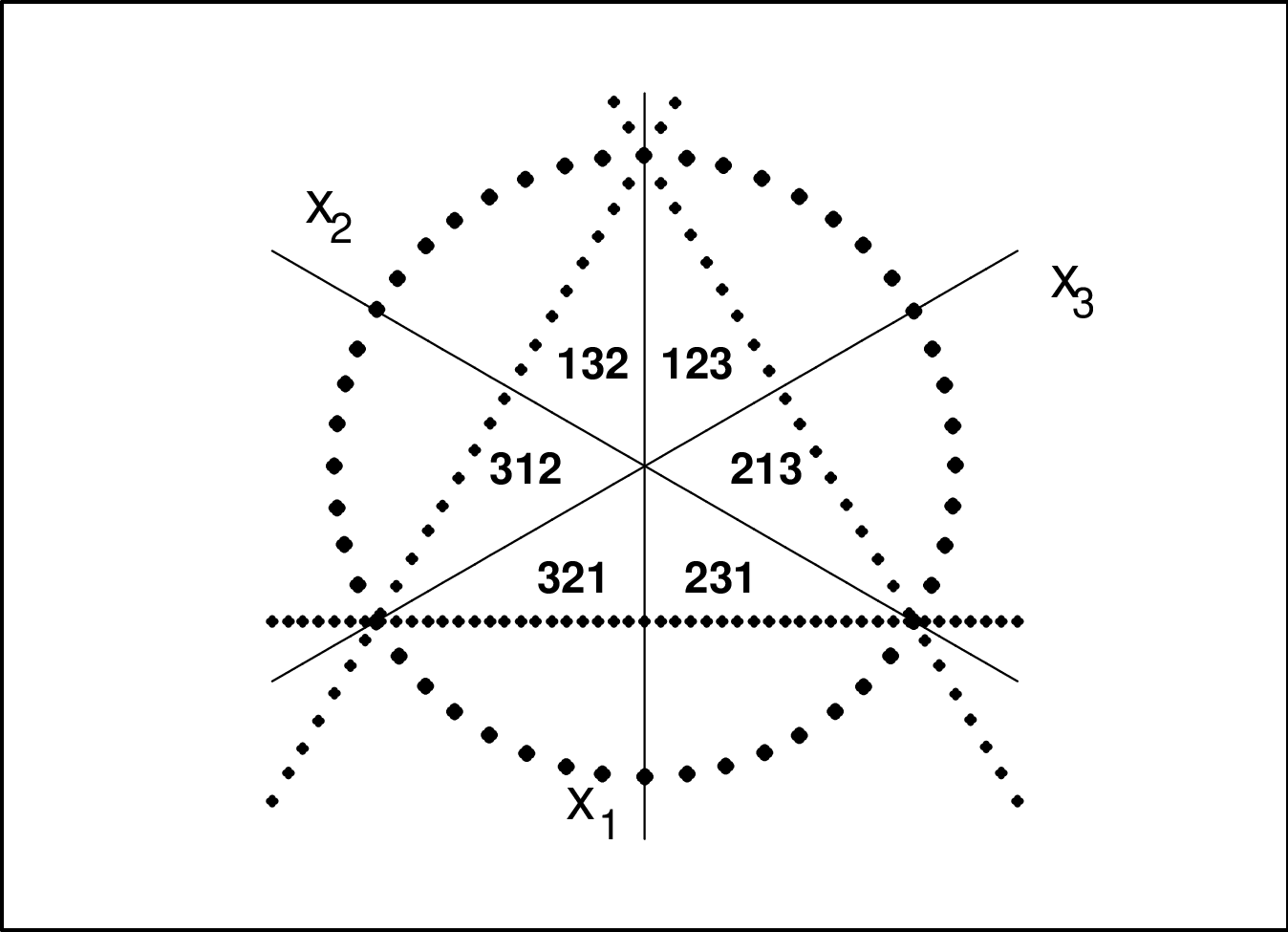,angle=270,width=0.4\textwidth}
\end{center}                         
\vspace{-2mm}\caption{
Weyl-chamber projections on the eye-guiding
circle and triangle.\label{3ww}}
\end{figure}

Another, third,
equally
$(A=3)-$specific but
formally independent
definition of the Weyl chambers
can be based on
a spherical-coordinate reparametrization
$X = \rho\,\sin \phi$ and $Y =
\rho\,\cos \phi$ of the plane.
After such a change of perspective our Calogero-Schr\"{o}dinger
Eq.~(\ref{gular}) is found to be solvable by the separation
of variables. This appeared to be
a decisive discovery in \cite{Calogero}
where the ansatz $\Phi(X,Y)= \psi(\rho) \chi(\phi)$
led to angular equation
 \be
\left ( -\frac{d^2}{d \phi^2} + \frac{9\,C}{\sin^23\,\phi}
\right ) \chi_k(\phi) = \varepsilon_k\,\chi_k(\phi)
\,,\ \ \ k = 0, 1, \ldots\,.
 \label{angular}
 \ee
The
impenetrable
barriers
of Figure \ref{3wwu}
are now localized simply
as lying
along the constant-angle lines
with $\sin 3\,\phi_{critical}=0$.

The split of the two-dimensional plane into six Weyl chambers
$W_{\{ijk\}}$ is presented, in a slightly modified manner, in Figure
\ref{3ww}. The three solid lines of the preceding picture are
reinterpreted as separating those subscripts of the Weyl chambers
(i.e., those triplets of integers $\{ijk\}$) which only differ by a
single elementary transposition. {The} line which represents the
two-body repulsion barrier $\sim (x_1-x_2)^{-2}$ and which stands
for the transposition $1 \leftrightarrow 2$ is marked now by symbol
``$x_3$'', etc.

\subsection{Non-equal barriers vs. loss of solvability}


By far the most important formal merit of
the conventional symmetric
(i.e., equal-coupling) Calogero-Schr\"{o}dinger Eq.~(\ref{gular})
is that after the separation of variables,
both the angular Eq.~(\ref{angular})
and its more common harmonic-oscillator radial-equation partner
prove solvable in closed form. This yields the spectrum
 \be
E=E^{}_{n,k}=\sqrt{\frac{3}{8}}\omega\,(4n+6k+ 6\alpha + 5)
,\ \ \ \  \ \ \ \alpha=\frac{1}{2}\sqrt{1+4C}>0, \ \ \ \ \ n,k=0,
1, \ldots\,
\label{dots}
 \ee
exhibiting multiple accidental degeneracies plus, in addition,
a global sextuple degeneracy
reflecting the split of the $X-Y$ plane in the six
dynamically independent Weyl chambers.
In this sense, every energy level of Eq.~(\ref{dots})
can be assigned six independent wave functions.
Due to the impenetrability
of the barriers, every one of these functions
may be chosen as exclusively
supported by one of the eligible Weyl chambers,
 \be
 \Phi=\Phi_{n,k}^{\{ijk\}}(X,Y)\,\,
 \left \{
 \begin{array}{ll}
 \neq 0\,\ \ {\rm for}  & (X,Y) \in W_{\{ijk\}}
 \\
 = 0\,\  & {\rm otherwise}\,. \ \
 \ea
 \right .
 \label{sedu}
 \ee
In opposite direction one could argue that
the admissibility of the choice (\ref{sedu})
(where the index $ijk$ runs over all six permutations of triplet $123$)
is a consequence of the impenetrability of the Weyl-chamber
boundaries. Implying that every one of the levels $E^{}_{n,k}$
of Eq.~(\ref{dots}), ``locally'' (i.e., inside $W_a$) degenerate or not,
is also six times degenerate, in addition, ``globally'' (i.e., inside $\mathbb{R}^2$).

Let us now remind the readers that
our present project is
aimed, first of all, at
an unfolding of the ``global'' spectral degeneracy.
The goal is to be achieved
by means of
an asymmetrization of the
Hamiltonian.
In the
$A=3$ special case
one of our guiding ideas is that
the strength of the singularities of Hamiltonian ~(\ref{CaHa})
is only controlled
by a single coupling constant $C$.

Preliminarily, one of the most natural tentative generalizations
of model
(\ref{gular}) might be sought via an {\it ad hoc\,}
replacement
of the single coupling constant $C$
by a
triplet of independent parameters,
 \be
 \frac{C}{(\sqrt{3}Y-X)^{2}} \ \to \
 \frac{Z_1}{(\sqrt{3}Y-X)^{2}}
 \,,
 \ \ \
 \frac{C}{(\sqrt{3}Y+X)^{2}} \ \to \
 \frac{Z_2}{(\sqrt{3}Y+X)^{2}}\,,
 \ \ \
 \frac{C}{X^2} \ \to \ \frac{Z_{3}}{X^2}\,.
  \label{zinq3}
 \ee
The denominators in (\ref{zinq3}) and, hence,
also the localization of
all of
the Weyl-chamber boundaries would remain the same.
In combination
with
the loss of symmetry,
their impenetrability
would imply
the loss of the coincidence of the six
Weyl-chamber-related
sub-spectra $\{E^{{\{ijk\}}}_{n,k}\}$.

As long as the
potentials inside the chambers can be now
treated as a set of independent, isolated
two-dimensional quantum dots,
only the accidental ``local'' degeneracy
between the
bound-state sub-spectra
will survive.

\begin{rem}
\label{mezilemma} The sextuple degeneracy of the energy levels
(\ref{dots}) can be unfolded using the three-parametric
asymmetrization (\ref{zinq3}) of the reduced Calogero-like
Schr\"{o}dinger Eq.~(\ref{gular}). {In such a case, the standard
proof of the exact solvability of the model will not apply. In the
future, indeed, the system might still be found solvable. On the
present level of knowledge, unfortunately, the full exact
solvability has to be declared, with all probability, lost.}
\end{rem}
{Incidentally, the latter methodical uncertainty may be related,
perhaps, to the well known fact that in several generalized
Calogero-type models the not quite expected exact solvability has
been found as restricted to a {\em finite\,} subset of certain
anomalous, ``anharmonic'' bound states \cite{qesa}. A connection of
such an incomplete, ``quasi-exact'' form of solvability with our
present approach is nontrivial. Although the possible comparisons
already lie beyond the scope of our present study, interested
readers may find more details and inspiration, say, in the
comprehensive review paper \cite{qesb} published, in 2005, as a part
of the Birthday Issue dedicated to Francesco Calogero on the
occasion of his 70th birthday.}

Another observation {related to Remark \ref{mezilemma}} is that up
to an accidental occurrence of a degeneracy, and in a sharp contrast
{to} the symmetric scenario, virtually all of the bound states of
the asymmetric model with unfolded spectrum will {\em necessarily\,}
have to obey the single-chamber-support restriction (\ref{sedu})
imposed upon{} wave functions.

\subsection{Reinstallation of solvability: Asymmetric barriers}


The validity and constructive nature of Theorem \ref{preteor} at $A=2$
might and have to be extended to
all of the models with $A>2$. In particular,
the degeneracy of the Calogero's $A=3$ spectrum has to be
unfolded. Obviously, this degeneracy can be
attributed to the
particle-permutation symmetry of the Hamiltonian. Thus, as long as we wish to
weaken the degeneracy
in a way different from the preceding tentative
recipe~(\ref{zinq3}),
we have to find another method of breaking that symmetry.

Naturally, we wish to do so
without the loss of the exact, closed-form solvability of
the original symmetric model.
The overall strategy of so-doing
has already been indicated above.
We have to pick up
six different coupling strengths $C_{\{ijk\}}$
and introduce
six reparametrized
versions
of
the Calogero's exactly solvable
Hamiltonian of Eq.~(\ref{CaHa}),
$H \to H_{\{ijk\}}=H^{(3)}(\omega,C_{\{ijk\}})$.
Now, although all of these operators
are defined over the whole, unrestricted
range $\mathbb{R}^2$ range of the coordinates,
the impenetrability of the barriers enables us
restrict their respective actions
to the functions over a single Weyl chamber.

In this manner, we managed to complement
the properly asymmetrized $A=2$ model
(\ref{ainq2})
by its exactly solvable $A=3$ descendant.

\begin{thm}
\label{3theo}
Calogero-like quantum system with Hamiltonian
 \be
 H^{(3)}(\omega,\vec{C}) = \bigoplus_{\{ijk\}}
 H^{(3)}(\omega,{C_{\{ijk\}}}) \,
 \label{devatenact}
 \ee
where
 \be
 \ \ \ \
 H^{(3)}(\omega,C_{\{pqr\}}) = - \sum_{m=1}^{3}\
 \frac{\p^2}{\p{x_m}^2}
 + \sum_{ m<n=2}^{3}\left [
\frac{1}{8}\,\omega^2   \,(x_m-x_n)^2
 +\frac{2C_{\{pqr\}}}{
 (x_m-x_n)^{2}}\right ]\,,\ \ \ \
 x_p<x_q<x_r\,.
  \label{ainq3}
 \ee
is exactly solvable.
\end{thm}
\begin{proof}
Such a model is a direct sum over all six permutations of
integers $\{123\}$. Every one of its
independent components (\ref{ainq3})
is obtained from the conventional model (\ref{CaHa})
after the elimination of the center-of-mass degree of freedom.
Hence, up to the loss of the ``global'' degeneracy,
the bound-state solutions remain ``locally'' unchanged
even after
the range of coordinates gets restricted from
${\mathbb{R}^{A-1}}$ to $W_{\{pqr\}}$,
and after the coupling is
chosen equal to $C_{\{pqr\}}$.
\end{proof}
Six independent
coupling constants $C_{\{pqr\}}$ in Eq.~(\ref{ainq3})
({\it alias\,} the ``colors'' of the chambers)
may but need not be all mutually different. Serving for
an exhaustive classification
of all of the asymmetric generalizations of the symmetric
model.
The classification can be given the form of a
list of all of the possible nontrivial colorings
of the six wedges in Figures \ref{3wwu} or \ref{3ww}.

As long as the explicit formulation of the list of the colorings
is conceptually elementary, it
may be left to the readers.
Let us only note that
besides the original Calogero's fully symmetric case
(using just a single color $C$) and besides the
maximally asymmetrized
six-color system admitting a maximal reduction of
the degeneracy,
some of the other special colorings might also prove to be of
an enhanced interest in applications.
For example,
we could be interested in
the form of boundaries
between the neighboring Weyl chambers $W_a \subset \mathbb{R}^2$
and $W_b \subset \mathbb{R}^2$.
The reason is that these boundaries in $\mathbb{R}^2$
are just pull-downs
of the microscopic particle-particle
repulsion-force singularities in $\mathbb{R}^3$.
Various ``realistic'' sub-classifications might be then
motivated by the need of being given a phenomenological input knowledge of
dynamics on the microscopic level,
with the emphasis put upon the difference between
the
symmetric barriers (when $C_a=C_b$) and
their asymmetric repulsion-force alternatives
(when $C_a \neq C_b$).
Thus, in particular,
whenever only
some of the particle-particle
barriers remain symmetric,
we could speak about
hybrid models. In these cases, the unfolding of the ``global'' spectral
degeneracy would be just partial.

Alternatively one could demand that even when {\em all\,} of
the particle-particle barriers
become left-right
asymmetric,
the breakdown of symmetry
may be incomplete,
leading again
to the mere partial removal of the
model's spectral degeneracy.
Once we recall
Figure \ref{3ww} such an arrangement might be immediately visualized
via coloring(s) in which
one only uses a minimum (i.e., say, two) different colors.
In these cases,
the wave function of the
system can (i.e., are allowed to) remain non-vanishing
in more than one Weyl chamber, i.e., in principle,
in
the whole
union of all of the
Weyl chambers carrying the same color.

\section{Spectral degeneracy and its unfolding beyond $A=3$\label{4revisited}}

\subsection{Exact solvability}

The systems with different $A$s will
share formal analogies.
In all of them, in particular,
the Weyl chambers may be
numbered
by the configurations {\it alias\,} permutations of the particles.
We will abbreviate
$a=\{i_1,i_2,\ldots,i_A\}$ and denote $W_a$ if and only if
$x_{i_1}<x_{i_2}<\ldots <x_{i_A}$.
Next,
the
singularities
will occur, in $\mathbb{R}^A$, whenever
$x_i = x_j$.
After
we eliminate the center of mass \cite{haka0},
these singularities will
cut and split also the reduced space $\mathbb{R}^{A-1}$
into its
Weyl-chamber subsets $W_{a}$.
The  points of their boundaries  $\p W_{a}$
are
precisely the points at which
the microscopic particle-particle interaction
becomes singular.

For all of these reasons, many
above-mentioned $A=3$ considerations can be generalized to
any $A$. In particular,
the generalizations of
Eq.~(\ref{devatenact}) and of
Theorem \ref{3theo} to any $A$ are immediate.
\begin{thm}
\label{Atheo}
At any integer $A \geq 2$, the Calogero-like quantum system with Hamiltonian
 \be
 H^{(A)}(\omega,\vec{C}) = \bigoplus_{\{i_1i_2\ldots i_A\}}
 H^{(A)}(\omega,{C_{\{i_1i_2\ldots i_A\}}}) \,
 \label{Adevatenact}
 \ee
where
 \ben
 H^{(A)}(\omega,C_{\{p_1p_2\ldots p_A\}}) =
 \ \ \ \
 \ \ \ \
 \ \ \ \
 \ \ \ \
 \ \ \ \
 \ \ \ \
 \ \ \ \
 \ \ \ \
 \ \ \ \
 \ \ \ \
 \ \ \ \
 \een
 \be
 = - \sum_{m=1}^{A}\
 \frac{\p^2}{\p{x_m}^2}
 + \sum_{ m<n=2}^{A}\left [
\frac{1}{8}\,\omega^2   \,(x_m-x_n)^2
 +\frac{2C_{\{p_1p_2\ldots p_A\}}}{
 (x_m-x_n)^{2}}\right ]\,,\ \ \ \
 x_{p_1}<x_{p_2}<\ldots <x_{p_A}\,
  \label{ainqA}
 \ee
is exactly solvable.
\end{thm}
At any $A$, in other words, the ultimate
asymmetrized Calogero-like Hamiltonian
remains
formally defined as a direct sum
of its mutually non-interacting
particle-ordering-dependent components.
Every such a model
remains solvable. Even the
limiting transition to the symmetric special case will be smooth.

The only remaining open problem
seems to be the formulation of an explicit list of
all of the structurally non-equivalent colorings $C_{\{p_1p_2\ldots p_A\}}$.
Thus, the qualitative characterization
of the set of the repulsive-barrier coupling
constants
is a crucial task.
Obviously, the task requires just a manageable
visualization
of the
$(A-1)-$dimensional space of the relative coordinates
after its decomposition
${\mathbb{R}}^{A-1}=\bigcup\,W_{a}$.

In such a framework, we are now prepared to
give
our preceding result
an alternative formulation.
\begin{thm}
\label{theoremone} At any $A$, the exact solvability of the
symmetric or asymmetric Calogero-like model of Eqs.~(\ref{CaHa}) or
(\ref{Adevatenact}) can be achieved via an explicit specification of
the couplings $C=C_{a}$ which must {only} be the same at all of the
inner sides of the boundaries of $W_a$.
\end{thm}
This observation provides the
background and reason why
our preceding constructive analyses
are instructive and
have to be complemented
by their extension beyond $A=3$.

In our ultimate illustration
the choice of $A=4$ does really form a certain
``feasibility bridge'' connecting
the easy models where $A!\leq 6$
with the overcomplicated ones where
$A! \geq 120$.
Thus, it makes sense to show
that also the four-particle Calogero-like asymmetrized models
of Eq.~(\ref{Adevatenact})
admit
a feasible, compact and more or less explicit
descriptive analysis.

\subsection{Weyl-chamber boundaries and their cartography at $A=4$\label{fejed}}

One of the key formal advantages
of the conventional as well as generalized
Calogero-like models with $A\leq 3$
is the separability
of the initial partial differential
Schr\"{o}dinger equation,
i.e., its reducibility
to a set of ordinary differential
eigenvalue problems.
Such a feature is lost at $A\geq 4$
so that the simplest elements of this class with $A=4$
deserve our constructive attention.
Moreover,
in a way
advocated in the literature the choice of $A\leq 3$
as made in the two preceding sections
need not
be sufficiently representative, in a broader methodical context
at least \cite{turbinerc,turbinerf}. As we already mentioned,
the enumeration of the colorings in Figures \ref{3wwu} or \ref{3ww}
is next to trivial.
We found it desirable to
extend our discussion beyond $A=3$, therefore.

\subsubsection{Colorings  at $A=4$.}

At $A=3$ the task of the split
of the plane $\mathbb{R}^2$ by the
sextuplet of the wedge-shaped Weyl chambers
as well as the related classification {\it alias\,} coloring
problems
were comparatively easy to settle. {\em All\,} of the
half-line Weyl-chamber
boundaries shared a common vertex
and
formed a planar star. This made
the coloring of the wedges $W_{\{ijk\}}$ easy, leading immediately
to an exhaustive
classification of all of
the unfoldings of the degeneracies, i.e., of
all of the spectrally non-equivalent models.
Moreover, it appeared sufficient to classify just
the colorings of
the six segments of an auxiliary circle
or, equivalently, of an auxiliary triangle (cf. Figure \ref{3ww}).

{The} situation at $A=4$ is perceivably more complicated. Indeed, we
have to replace, first of all, the $A=3$ change of variables
(\ref{ular}) by its $A=4$ analogue. Nevertheless, the details of the
$A=4$ change of variables as well as an explicit discussion of a
deep algebraic discussion of its consequences will be skipped here,
for three reasons. Besides the first one (viz., the sake of brevity:
their form is well known) and the second one (viz., a redundancy for
our present purposes: the algebraic $A=4$ formulae already become
far from transparent), our third and main reason is that in a way
sampled by Figure \ref{3ww} at $A=3$, it will be {\em sufficient\,}
to know only the topology (i.e., the neighborhoods) of the $A=4$
Weyl chambers. {\em Just\,} in its graphical form.

%
%
%
%
%
%
%
%
%

\subsubsection{Projection on a central sphere.}

Interested readers may find a deeper insight in the formulae and
multiple deep algebraic structures in the dedicated literature
\cite{turbiner,coxeter,coxeterb}. At $A=4$, in particular, one
suddenly has to deal with as many as $A!=24$ eligible Weyl chambers.
Thus, a replacement of the formulae by their graphical alternatives
seems well founded. It enhances, indeed, the efficiency of the
analysis. In particular, the description becomes facilitated when we
recall and adapt the graphical tricks as used at $A=3$. In the first
step, the centralized dotted-curve auxiliary planar circle of Figure
\ref{3ww} can be replaced,  at $A=4$, by a sphere in three
dimensions.

Interested readers are encouraged to
find a nice picture of such a sphere on the web \cite{dek}.
In this representation one can also
keep the analogy with the planar case of Figure
\ref{3ww} and project the
three-dimensional objects $W_{\{ijkl\}}$ on the spherical triangles
as displayed on the sphere in \cite{dek}.

\subsubsection{Projections on a central cube.}

The role of the
dotted-line triangle of Figure \ref{3ww}
can readily be transferred, at $A=4$,
to a central tetrakishexahedron
of Ref. \cite{dekb}.
A key advantage of such an upgrade of projection is that
it further simplifies the representation.
It
replaces the
projections
of the three-dimensional Weyl chamber pyramids $W_{\{ijkl\}}$
on the spherical triangles
by the flat,
planar triangles of the surface of the tetrakishexahedron
of Ref. \cite{dekb}.
%
%
%
%
%
%
%
%
%
%
%

Although the
colorings of the planar triangular faces of
the tetrakishexahedron
look already feasible,
let us still
introduce another, last simplification
converting the tetrakishexahedron  of Ref.~\cite{dekb}
into a cube, with its six faces divided into quadruplets of
the neighboring $W_a-$representing triangles.
For completeness,
interested readers can find
the three-dimensional
picture of such a cube with subdivided faces
on internet~\cite{dekg}. In Figure \ref{4ww}
we present here its planar, two-dimensional unfold.

\begin{figure}[h]                     
\begin{center}                         
\epsfig{file=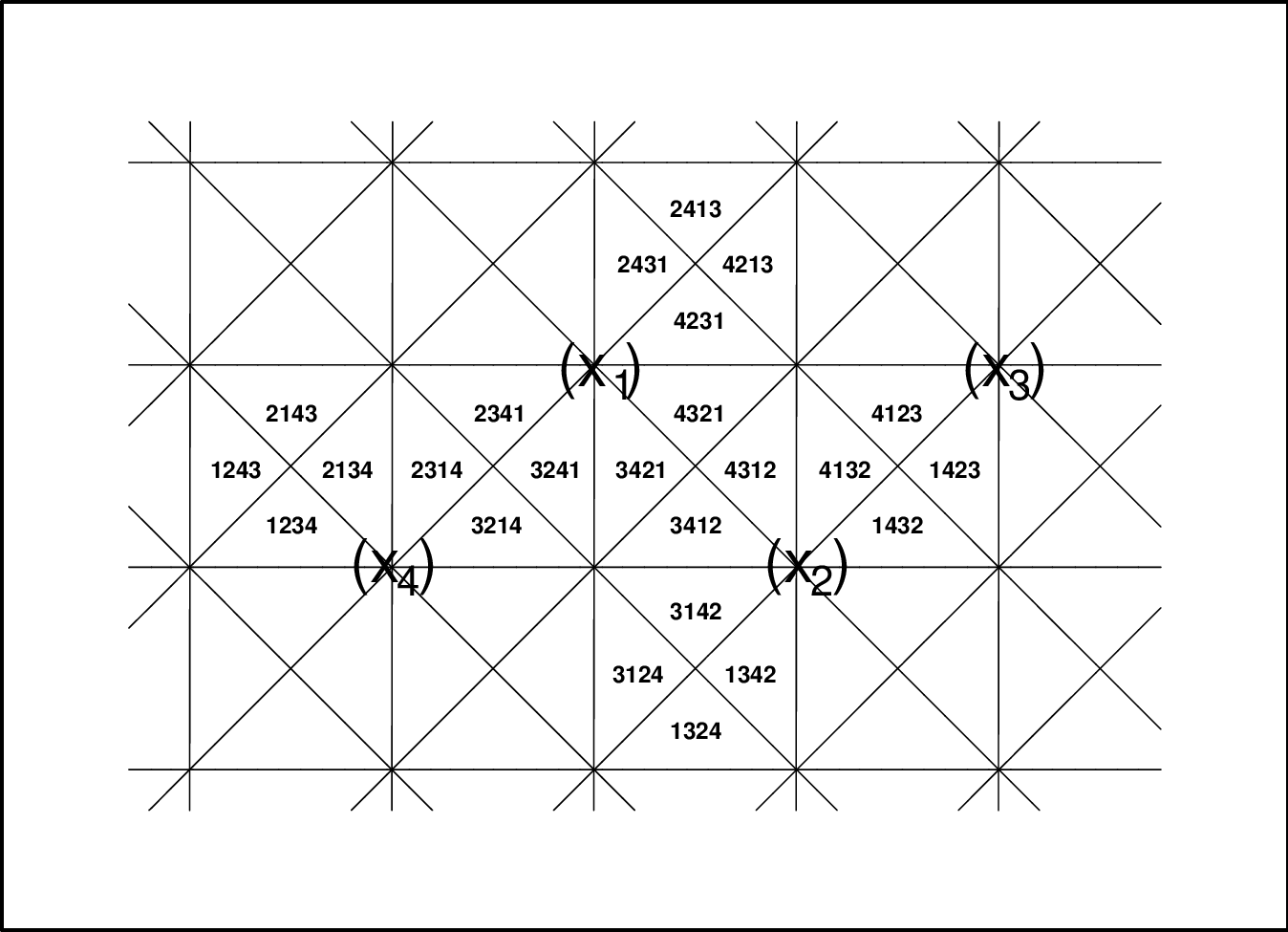,angle=270,width=0.7\textwidth}
\end{center}                         
\vspace{-2mm}\caption{Twenty four Weyl chambers $W_{\{ijkl\}}$
in projection on a central cube.
 \label{4ww}}
\end{figure}

\subsubsection{Numbering by permutations.}

For our present purposes we may cut and {display} the surface of the
cube in the planar-cartography representation of Figure~\ref{4ww}.
We only have to remember that such a planar picture might require,
for some purposes, a reconstruction of the three-dimensional surface
of the cube by bending and gluing some of the edges.

The union
${\mathbb{R}}^{A-1}=\bigcup\,W_{a}$
has to be taken over all of the
permutations of the particles.
It is only important for us to know that
{\em before\,} the projection as sampled in Figure \ref{4ww}
the geometric shape of the individual three-dimensional $A=4$
Weyl chambers $W_{\{ijkl\}}$
will have again the pyramidal form, with their top fixed in the origin.
We may expect that in comparison with the three-particle model,
the generalization of the ``cartography'' to $A=4$
(and, in principle, to any larger $A>4$)
is straightforward.

After the projection on the cube
we have to deal with the
planar representation of the surface.
The
vertices of the cube coincide with
the intersections of the six
singularity lines
(i.e., of the three lines standing for the adjacent edges and
of the other three ones
for the face diagonals).
One only has to remember that these vertices
were defined as the octuplet of points
in $\mathbb{R}^3$ at which the circumscribed
auxiliary sphere is intersected by the four straight lines
representing
the original particle-coordinate axes of $x_1$, $x_2$, $x_3$ and $x_4$.
These intersections are also identified in Figure \ref{4ww}.

{In} this picture the surface of the inscribed cube is represented
by the six concatenated squares. The symbols $(x_j)$ mark the
intersections of{} outwards-running axes with the sphere and with
the cube (at its vertices). The planar triangular projections of all
of the twenty four three-dimensional Weyl chambers $W_{\{ijkl\}}$ on
the walls of the cube are finally numbered here by the quadruplets
$ijkl $ of integers such that $-\infty < x_i<x_j<x_k<x_l< \infty$.
In this manner we obtained again a transparent schematic
two-dimensional representation of the Calogero-model kinematics at
$A=4$.

In the same picture the straight lines mark the singularities (i.e.,
the barriers) and the numbers {$ijkl$} characterize the (projections
of the) $A=4$ chambers $W_{\{ijkl\}}$. One can easily check that
whenever one crosses the line, the quadruplet of integers {$ijkl$}
is changed {just} by an elementary permutation of its neighboring
elements. This permutation (i.e., a re-arrangement of the particles)
can be perceived as a consequence of the crossing of the barrier
{\it alias\,} singularity line.

As a serendipitious byproduct of these topological considerations
we may deduce the following, not quite expected
combinatorial result.

\begin{lemma}
\label{lemmajedna}
Once we
localize the positive
half-line intersections $x_1$, $x_2$, $x_3$ and $x_4$
on
the auxiliary cube's surface,
the allocation of the
permutations {\rm ijkl} to the triangles of the cube's surface is
unique.
\end{lemma}
This observation returns us to the model with $A=3$
where we marked, in
Figure \ref{3ww}, also the projections of
the triplet of the original Cartesian particle-coordinate axes
$x_1$, $x_2$ and $x_3$.
These lines intersected the auxiliary circle at the six
``barrier-sampling'' points.
At $A=4$ the idea remains applicable because now,
the intersections of the sphere
with the axes $x_1$, $x_2$, $x_3$ and $x_4$ form
an octuplet of the ``marked''
points in the diagram of Figure \ref{4ww}.

We may identify the latter points with
the vertices of our auxiliary cube defined as inscribed in the
auxiliary sphere.
This enables us to
project the
boundaries of the separate Weyl chambers on the walls of the cube
yielding the singularity lines in Figure \ref{4ww}.
Summarizing,
the
twenty four
Weyl chambers are mapped on the set of
spherical triangles
covering the auxiliary sphere.
They are, incidentally, rectangular.
Their planar projections
cover the surface of the above-mentioned
inscribed cube, with the rectangularity preserved.
What is important is that the one-to-one correspondence
between the topology of
the three-dimensional
Weyl chambers
and the topology of the
``numbered'' triangles in Figure \ref{4ww} 
is
achieved and guaranteed.

\begin{figure}[h]                     
\begin{center}                         
\epsfig{file=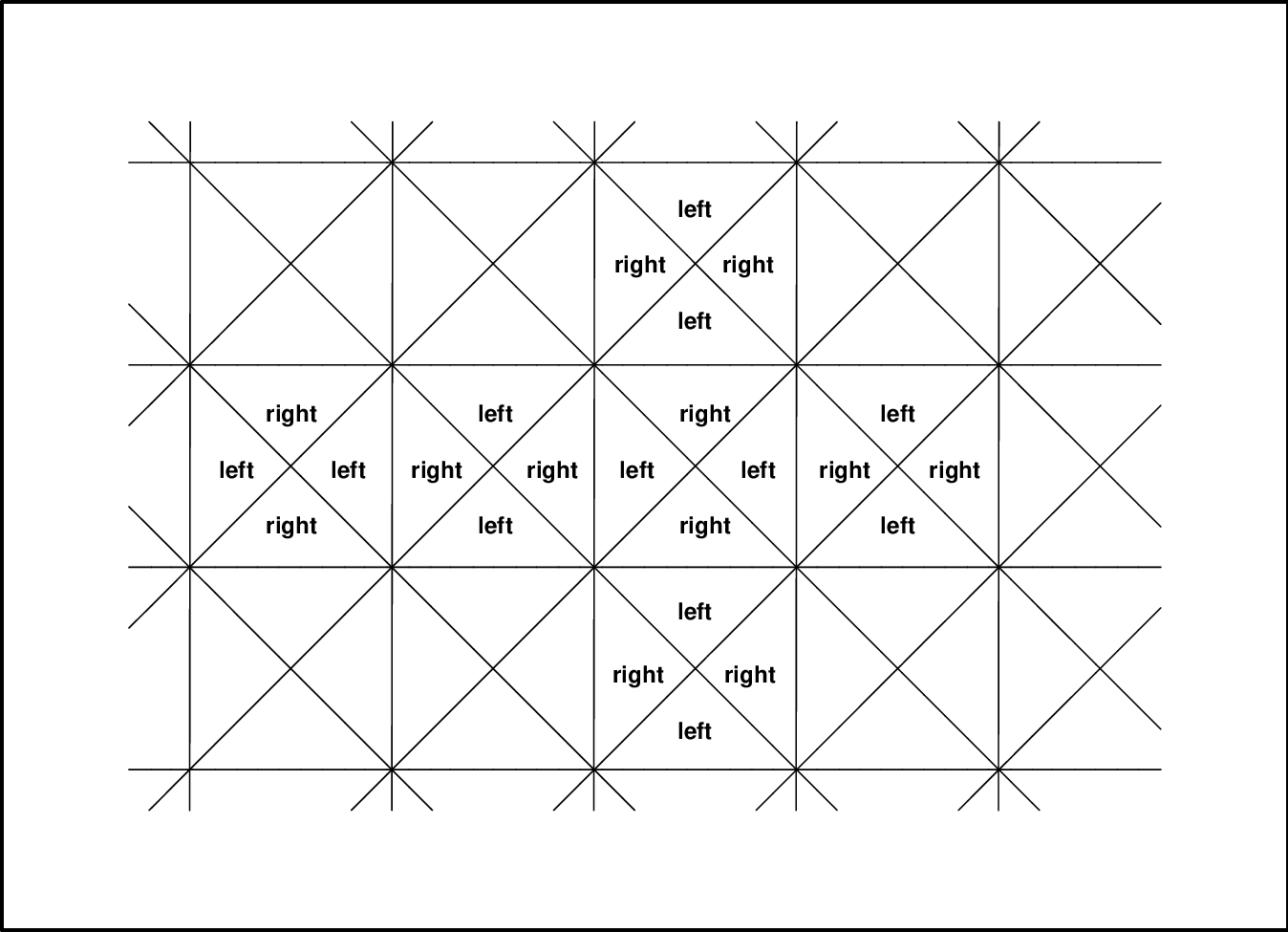,angle=270,width=0.7\textwidth}
\end{center}                         
\vspace{-2mm}\caption{The $A=4$ coloring
with unique (viz., left-right)
asymmetry of {\em all\,} of the barriers.
 \label{col4ww}}
\end{figure}

\subsection{Alternative cartography at $A=4$\label{fexjed}}


Let us repeat that
after a routine elimination of the
center-of-mass degree of freedom
the $A=4$ constructions
have to be simplified
by the reduction of the initial partial
differential Schr\"{o}dinger equation
in four coordinates
to its three-dimensional
effective version \cite{haka4}.
As long as we are mainly
interested in the structure
of the
decomposition of ${\mathbb{R}}^3$
into a union of the Weyl chambers,
we will not need the explicit
formulae for the wave functions.
These formulae
may already be fairly complicated, see \cite{turbiner,Brink}.
Fortunately, our present task is just the specification of the
Weyl chamber or chambers $W_a$ in which
the wave function of
a relevant bound state
can be nontrivial, $\psi_a \neq 0$.

Such a specification should be based on a visualization of the
decomposition of ${\mathbb{R}}^3$ and on the phenomenologically
motivated assignment of the equal or different couplings $C=C_a$ to
the respective chambers $W=W_a$ at all of the position permutations
$a=\{i_1,i_2,i_3,i_4\}$. Naturally, an exhaustive and systematic
account of all of the possibilities would already be too long. For a
selection of a sample we propose to combine a ``minimal'' nontrivial
coloring (using just two colors $C=C_{left}\neq C_{right}$) with a
``maximal'' asymmetry of the barriers (meaning that none of the
barriers $\sim (x_i-x_j)^{-2}$ remains symmetric, i.e., invariant
under the exchange of particles $x_i \longleftrightarrow x_j$). {A}
sample of the maximally asymmetric arrangement (using just two
colors called ``left'' and ``right'') is displayed in Figure
\ref{col4ww}.

\subsubsection{Projections of the set of Weyl chambers on
tetrahedron}

\begin{figure}[h]                     
\begin{center}                         
\epsfig{file=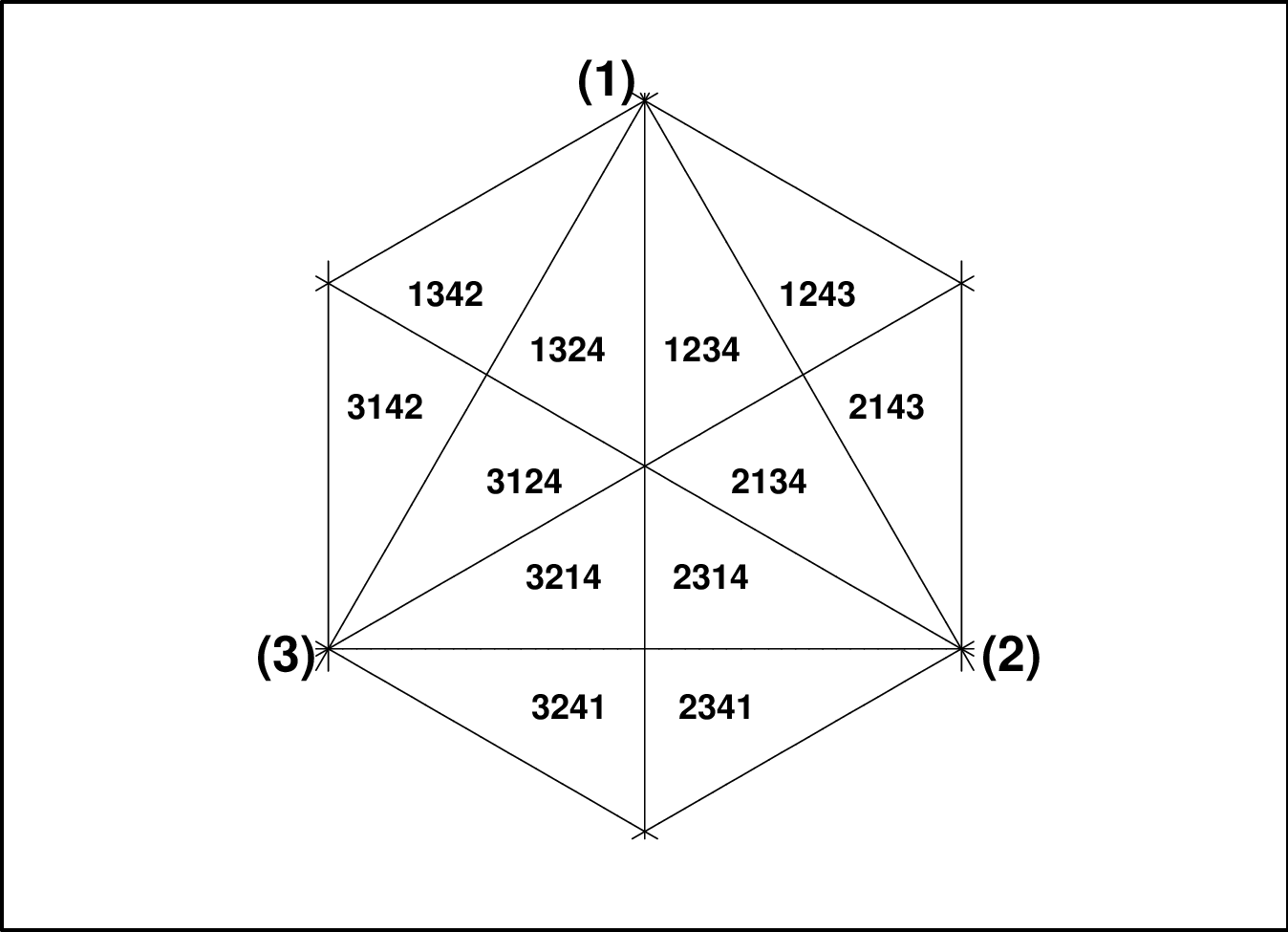,angle=270,width=0.7\textwidth}
\end{center}                         
\vspace{-2mm}\caption{The upper half of the auxiliary cube of Figure
\ref{4ww} as seen from vertex ``$(x_4)$''.
 \label{kol4ww}}
\end{figure}
%

 \noindent
{The} surface of the cube of preceding paragraph is a
three-dimensional object. {Its} cut and straightening {were} used
yielding Figures \ref{4ww} and \ref{col4ww}. An alternative approach
to the visualization of the same three-dimensional cube is sampled
in Figure \ref{kol4ww} where just half of its surface is presented
as seen after the cube is rotated in such a way that the vertex
marked by symbol ``$(x_4)$'' appears in front and in the very center
of the picture.

Although the latter planar view represents the three-dimensional cube,
we might re-read the same picture as a flat triangle
with vertices $(1)$, $(2)$ and $(3)$, endowed with the three
planar extensions
which could be bent along the edges and turned down.
A new
three-dimensional surface emerges.
Its complementary, ``invisible'' lower half
could be
reconstructed, given the same shape and glued to the upper half.


\begin{figure}[h]                     
\begin{center}                         
\epsfig{file=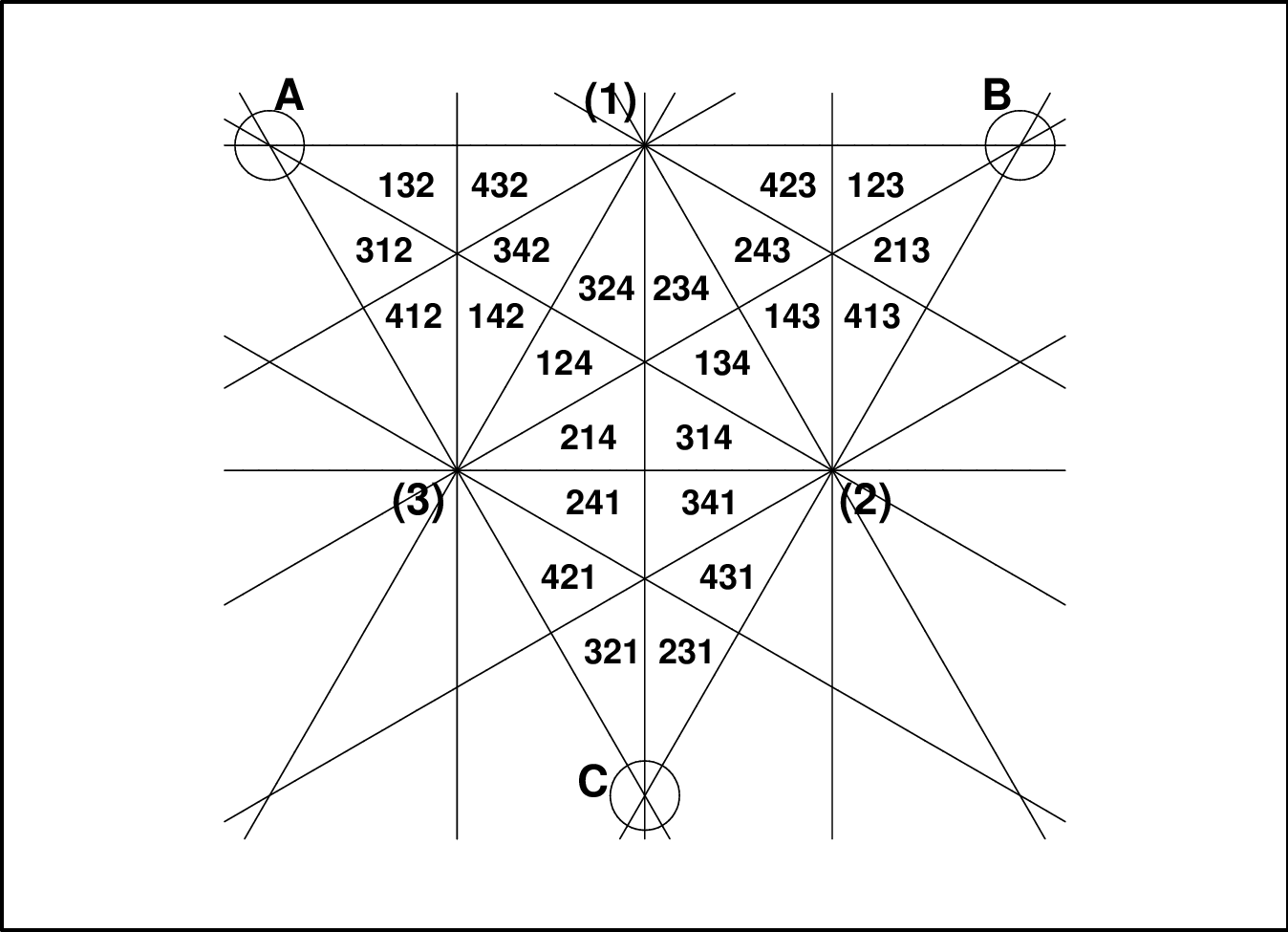,angle=270,width=0.7\textwidth}
\end{center}                         
\vspace{-2mm}\caption{A tetrahedral alternative to Figure \ref{4ww}
(the surface unfolded
into triangle $ABC$).
 \label{4xww}}
\end{figure}

Admissibly, such a type of reconstruction of the missing part of
Figure \ref{kol4ww} is oversophisticated. Nevertheless, just its
minor reinterpretation yields another, much more regular
Platonic-body representation of the topology, i.e., of the
representation of the Weyl-chamber neighborhoods. Explicitly it is
presented in Figure \ref{4xww}. The Weyl chambers (numbered there by
the mere triplets of integers) are organized there into three
triangles. {The} upper one (with vertices $(1)$, $(2)$ and $(3)$)
has to stay unchanged while the other three triangles have to be all
bent down, with their respective circle-marked ``outer'' vertices
$A$, $B$ and $C$ to be glued together. Forming the fourth vertex of
a tetrahedron which should be endowed with symbol $(4)$.

\begin{figure}[h]                     
\begin{center}                         
\epsfig{file=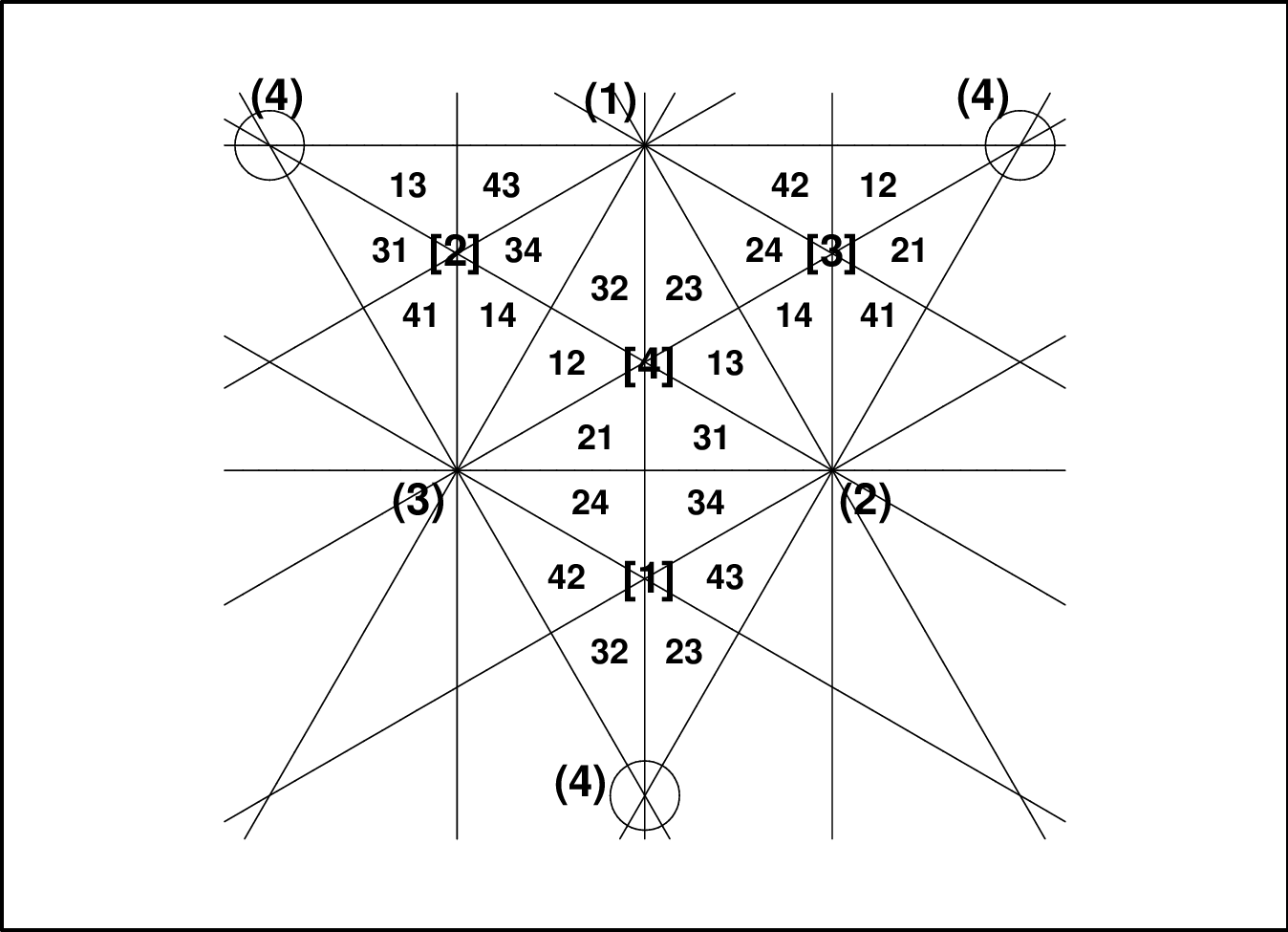,angle=270,width=0.7\textwidth}
\end{center}                         
\vspace{-2mm}\caption{An alternative version of Figure \ref{4xww}
(a shortened Weyl-chamber classification).
 \label{4yww}}
\end{figure}

In Figure \ref{4xww} the three
intersections marked as $(1)$, $(2)$ and $(3)$ have to be
interpreted as
the three vertices of one of the
four triangular sides of an auxiliary central tetrahedron.
Its remaining three sides are triangles $(1)B(2)$,
$(2)C(3)$ and $(3)A(1))$.
In the picture the surface of
the tetrahedron is unfolded so that after its refolding
(in three dimensions)
the three
``outermost'' vertices of the large triangle
(marked by the small circles and by the letters
$A$, $B$ and $C$)
will coincide.

{In} a marginal remark let us add that  for the sake of brevity we
numbered all of the small, Weyl-chamber-representing triangles just
by the triplets of integers $jkl$. The reason is that for a return
to the full-fledged notation $W_{\{ijkl\}}$ the ``missing'' value of
$i$ is provided by the nearest vertex $(i)$. {The} rule becomes
sufficient after the symbols $A$, $B$ and $C$ are all re-read as
$(4)$. Incidentally, in such a simplification we can move one step
further. This is explained in our last Figure \ref{4yww} in which we
attached the new quadruplet of symbols $[l]$ to the remaining four
vertices. Then, as long as all of the neighboring Weyl-chambers
$W_{\{ijkl\}}$ share the last integer $l$ with the vertex $[l]$,
just two digits suffice for the labeling.

\subsubsection{An ultimate classification of the barriers}

Besides a purely formal appeal of
the shortened notation. it also indicates, directly,
which particles are exchanged when crossing some of the singular
boundary lines. Due to this,
the combinatorics of coloring becomes easier,
simplifying
also our ultimate understanding of the
correspondence between the microscopic symmetries/asymmetries of the
particle-particle repulsion
and the structure of the macroscopic bound-state spectrum.

Any
phenomenologically motivated and exhaustive ``tetrahedral-surface''
classification of all of the 24 Weyl chambers
(and of their shared singular boundaries)
as provided by
Figures \ref{4xww} and \ref{4yww}  can be reinterpreted as a completion of
Figure \ref{kol4ww}, with its ``invisible'' half added.
The four-digit particle orderings $ijkl$
of Figure \ref{kol4ww} were shortened to the three-digit symbols $bcd$
in Figure \ref{4xww}. The reason is that
for the full identification of the chambers $W_{\{ijkl\}d}$ in our notation,
the knowledge of the trailing triplet $jkl$ is sufficient.
The ``missing'' value of $i$
(always equal to 1, 2, 3 or 4)
coincides with the one given in the parenthesized
symbol $(i)$ marking the adjacent vertex.
The validity of this coincidence can
immediately be checked by the inspection of the ``incomplete''
Figure~\ref{kol4ww}.

An alternative to Figure \ref{4xww} is presented in Figure \ref{4yww}
where  the remaining four ``anonymous''
sextuple intersections of the singularity lines
(localized in the centers of the sides of the tetrahedron)
are marked by the symbols  $[1]$, $[2]$, $[3]$ and $[4]$.
Again, having noticed that the adjacent chambers-marking triplets
of integers $jkl$ share the vertex $[z]$ with $z\equiv l$,
we introduced our ultimate index-shortening
convention $jkl \to jk\,$ in Figure~\ref{4yww}.

Under the latter convention one can finally decide to set the tetrahedron's
edges equal to one and to distinguish, subsequently, between the
``long'', ``middle'' and ``short''
Weyl-chamber singular boundaries of the respective lengths
$\alpha=1/\sqrt{3} \approx 0.577$,
$\beta=1/2$ and $\gamma=1/\sqrt{12} \approx 0.289$.
Using this notation one can conclude that among the total number of 36
singular boundaries we have to deal here with 12 short ones, 12 midle ones
and 12 long ones, with the long ones always having the
end points of the form $(i)[l]$ while representing the lines of
the reordering of the remaining two integers $j$ and $k$ in
our ultimate chamber-classification scheme of Figure \ref{4yww}.


\section{Summary\label{summary}}

One of the characteristic features of the conventional Calogero
model (\ref{CaHa}) which describes the one-dimensional motion of a
system of $A$ quantum particles is the symmetry of the Hamiltonian
with respect to the reordering of  {the} particles. Remarkably
enough, such a symmetry can be spontaneously broken because the
details of the dynamics are such that whenever one fixes the
ordering in advance, it will be, during the evolution, conserved.
{The} system can be perceived as composed of multiple (i.e., of as
many as $A!$) independent subsystems.

In such a situation it was natural to ask the question about the
feasibility of breaking the above-mentioned symmetry manifestly. The
motivation looked sound: In the symmetric model the spectrum was
multiply degenerate and it was not too clear how one could unfold
this degeneracy. At the same time, a promising methodical guidance
seemed to lie in a manifest violation of the symmetry in a way
illustrated by Figure \ref{picee3ww} {and by the ``validating''
small$-A$ examples as thoroughly described in sections
\ref{revisited} and \ref{3revisited}} above.

Along this line we found and described
one of the possible answers.
A modification of the model has been proposed
in which
the class of the
forces which control the evolution
has been extended.
The basic idea of the
innovation
may be seen in the dynamical nature of the
conservation of the ordering.
It appeared to be caused
by the original Calogero's particle-particle-repulsion forces
which were simply too strong. Their strength
made any
exchange of the neighboring particles prohibited.

{We} argued that the main reason of a complete suppression of the
rearrangements had to be seen, first of all, in the strongly
singular behavior of the Calgero's original repulsion forces at
short distances. From this we deduced that the whole rearrangement
suppression paradox can be identified as a mere misunderstanding. We
believe that  {the puzzle} is now clarified.

{We} recollected that the conventional interpretation of the
Calogero's quantum system (as currently accepted by the majority of
authors) is that its Hamiltonian $H$ is a direct sum of all of the
eligible ordering-dependent operators $H_a$ where the subscript
$a=\{i_1,i_2,\ldots,i_A\}$ denotes and defines such an ordering. We
imagined that from this point of view one can and should reinterpret
the singular particle-particle repulsion as an impenetrability of
the boundaries of the Weyl chambers $W_a$.

This was a decisive step.
It gave birth to the main idea of our present paper:
We proposed that one can construct a
fairly large family of
the new Calogero-like Hamiltonians
in which a central role is played by
the turn of our model-building attention
from the complete kinematical Euclidean space $\mathbb{R}^A$
(comprising all of the individual particle coordinates
$x_j \in \mathbb{R}$ with $j=1,2,\ldots,A$)
to its
decomposition into a union of the
{\it ad hoc\,} $A-$dimensional
submanifolds.

Routinely, after the standard and
trivial decoupling of the independent
center-of-mass motion,
the decomposition was reduced to the space $\mathbb{R}^{A-1}$
defined as a
union of the  $(A-1)-$dimensional
wedge-shaped Weyl chambers $W_a$.
Then, the rest and the largest part of the paper had to be
devoted to an explicit description and discussion
of some consequences of the main idea.

One of the phenomenologically most interesting consequences of the
decomposition of the set $\mathbb{R}^{A-1}=\bigcup W_a\,$
parametrized by the $(A-1)-$plet of relative coordinates
concerned the boundaries between the neighboring Weyl chambers.
These (repulsive) boundaries just correspond,
from the microscopic point of view, to the
mutual (repulsive)
particle-particle interactions. Thus,
in the language of mathematics,
both of these representations of
the
(impenetrable) boundaries (not admitting any tunneling)
appeared to be
allowed left-right asymmetric.

The
consequences of asymmetrization
of such an
``input'' information about dynamics
were discussed, in detail, at $A=2$, $A=3$ and
at $A=4$. In all of these
exemplifications the loss of symmetry was shown to lead to
a higher flexibility (which is due to the emergence
of many new and freely
variable parameters)
as well as to
a partial or complete suppression of the ``global'' degeneracy of
the spectrum of the original Calogero's model.

One of the most important challenges
appeared to be the preservation
of the exact solvability of the model after
its
asymmetrization.
Due to an (at least partial)
suppression of the degeneracy of
the spectrum of the conventional symmetric model
the related unfolding of the wave functions
was, really, expected to prove useful, say, during
some further future amendment of the model
using, say, perturbation theory.

In our text we pointed out that in the symmetric-model special case
the exact solvability
of quantum Hamiltonian~(\ref{CaHa})
is a nontrivial consequence of its
Lie-algebraic symmetry \cite{perelomov}.
The existence of this symmetry
has only been rendered possible by
the idealized picture of dynamics: The
motion of the multiplet of particles
is one-dimensional. Moreover,
the mutual two-body interactions between individual particles
are represented by the quartic plus
inverse quartic potentials
$V_{ij}^{(attractive)}\sim (x_i-x_j)^2$
and
$V_{ij}^{(repulsive)}\sim 1/(x_i-x_j)^2$,
respectively.
In our present paper, from this point of view,
the generalization can be simply
characterized as making
the shape of the forces $V_{ij}^{(repulsive)}$
asymmetric, i.e., different for $x_i>x_j$ and for $x_i<x_j$.
This rendered the new Calogero-like model multi-parametric.

Another important
aspect of the innovation is that in spite of the
emergence of new free parameters,
the resulting Calogero-like
model remained solvable exactly.
Although such a statement may sound surprising,
the explanation is not too difficult,
requiring just a more consequent use of the
concept of the Weyl chambers.
Incidentally, we could
borrow
their definition from the conventional Calogero's model.
Thus, what we only had to add was the
emphasis upon
the strongly singular nature of
the particle-particle short-range interactions
$V_{ij}^{(repulsive)}$.
We only had to reread the well known
impenetrability of these barriers
as
a ban of particle exchanges.

During the unitary time-evolution of our
generalized quantum system
the ordering of the particles can be perceived as conserved.
This is the most immediate consequence
of the full-scale suppression of the tunneling between
the neighboring Weyl chambers.
Thus, any bound state of any Calogero-like
quantum system can be understood as represented by the wave function
supported by a single Weyl-chamber subdomain of coordinates.
The only difference between the original and the present (i.e.,
``asymmetrized'') models is that in the former case, due to the
uniqueness of their coupling $C$, every bound-state energy
level is particle-ordering-independent.

We have shown how
the ``global'' degeneracy of the levels can systematically be removed.
What is sufficient is that in Hamiltonian (\ref{CaHa})
the coupling constant $C$
is made configuration-dependent.
Thus,
along all of the inner boundaries of one of the Weyl chambers
(say, of $W_a$
characterized by the arrangement
$x_{i_1}<x_{i_2}<\ldots <x_{i_A}$ of the particle
coordinates) we require that $C=C_a$
while in another Weyl chamber (say, in $W_b$ characterized by another
ordering $x_{j_1}<x_{j_2}<\ldots <x_{j_A}$
{\it alias\,}
permutation
of the particles) we set $C=C_b$, etc.
In general,
one then has to distinguish between two scenarios.
Either the respective singular coupling strengths coincide (i.e.,
we have $C=C_a=C_b$) or not.

In the latter case of our present interest
we had
to generalize, ``asymmetrize'' the form of the Hamiltonian of
Eq.~(\ref{CaHa}) by making it configuration-dependent. In particular,
once we wished to keep the model solvable,
we had to make {\em all\,} of the
choices of the
parameters
mutually compatible.
In this context, we found the criterion.
Naturally,
for an ultimate explicit construction of the
solvable multiparametric Calogero-like new model,
such a guarantee is a crucial mathematical ingredient.
So we showed that
such a guarantee is just a
combinatorial exercise, provided only that
we take into consideration the impenetrability of the
singular particles-separating barriers.

The absence of tunneling enabled us to
infer that
the conventional Calogero's choice of
the spatially symmetric {\it alias\,}
particle-exchange invariant
$V_{ij}^{(repulsive)}\sim 1/(x_i-x_j)^2$
is
not sufficiently general.
In a systematic search for its more general
spatially asymmetric
and configuration-dependent
alternatives we started our analysis at the
two-particle scenario. We replaced the
conventional Calogero's choice of repulsion by
the spatially asymmetric formula (\ref{ainq}) and outlined the consequences.
At $A=3$ we displayed
all of the six
wedge-shaped Weyl
chambers in Figure \ref{3wwu}.
In subsequent Figure \ref{3ww} we showed that and how
the two-dimensional
Weyl-chamber wedges
can be projected, for simplification purposes,
on the segments of an auxiliary circle or on the sides
of an auxiliary triangle.
The idea
has subsequently been transferred to
$A=4$.
We emphasized that the role of the
circle and triangle
at $A=3$ becomes inherited by a sphere and
a cube
or, alternatively, by a tetrahedron at $A=4$, etc.
This was
our last illustration and argument in support
of the consistency of the process of asymmetrization at any $A$.

Summarizing, we established, first of all, that what can
and should be asymmetric are
the boundaries between the Weyl chambers.
Secondly, we emphasized that
a return to the generic form
of the Calogero's Hamiltonian (\ref{CaHa})
is possible, provided only that this operator is
reinterpreted as particle-ordering-dependent.
Thirdly, its form of
a direct sum over $A!$ subsystems
has been shown to open the possibility of
an unfolding of the ``global''
spectral degeneracy.
Last but not least we pointed out that
due to the absence of tunneling, the unfolding of levels
can be
mediated simply by the choice of a suitable
multiplet of the particle-repulsion couplings $C_{\{i_1, \ldots,i_A\}}$,
leading to a classification via coloring
of the Weyl chambers.
For a really elementary illustration
a return to Figure \ref{picee3ww} can be, once more, recommended.

\newpage

\section*{Appendix: The occurrence of asymmetric
impenetrable barriers in relativistic quantum mechanics\label{discussion}}

The conventional non-relativistic
quantum theory admits,
in contrast to its classical limit, the
existence of
bound states even in a strongly singular attractive
potential
$V(x) = -g^2 X^{-2}$, provided only that
the coupling is not too large \cite{Landau}.
In paper \cite{Ishkha},
Ishkhanyan with Krainov
turned attention
to analogous singular-potential problems in
relativistic quantum mechanics.
Some of their results might find a
not quite expectable non-relativistic counterpart
in our present paper.

In their
study of a stationary one-dimensional Dirac equation
they considered, in particular,
a specific pseudoscalar screened Coulomb potential
exhibiting a conventional spatial symmetry,
 $$
 W^{(IK)}(x)\sim a\,|x|^{-1} + b\, |x|^{+1/3}\,.
 $$
In a way reported during recent conference \cite{Ishkhab}
they revealed that besides its well motivated phenomenological origin,
the model could also prove
methodically relevant because it can be
made conditionally exactly solvable (CES;
for a concise information
on the latter concept
see also our older comment \cite{PRAZnojil}).

From our present point of view
the essence of the
Ishkhanyan's and Krainov's message
is that
after one rewrites Dirac equation
in a mathematically equivalent
non-relativistic bound-state form,
one just has to solve
the effective
Schr\"{o}dinger equation
with a spatially manifestly asymmetric
potential
 \be
 V_{}^{(IK)}(x) = \left \{
 \begin{array}{lll}
 (x^2)^{1/3}+u_{Left}^{(IK)}/x^{2} \,,
 \ \ \ &u_{Left}^{(IK)}=\ell(\ell+1)>0\,,\ \ \ &x < 0\,,\\
 (x^2)^{1/3}+v/(x^2)^{1/3} +
 u_{Right}^{(IK)}/x^{2}
 \,,\ \ \ &u_{Right}^{(IK)}=-\ell(-\ell+1)<0\,,
 \ \ \ &x > 0\,
 \ea
 \right .
 \label{VIK}
 \ee
with parameter $\ell \in (0,1)$.
In particular, their model possessed a spatially manifestly asymmetric
central singularity with which
the mathematically desirable CES property
has been found to occur at a
special value of $\ell= \ell^{(IK)}=1/6$.

We have to add that
the authors of paper \cite{Ishkha} worked with an extreme,
phenomenologically ambitious asymmetry
of their potential. In our present notation
they used just
Eq.~(\ref{ainq}) with
two independent couplings,
 \be
 \ell_{(left)}(\ell_{(left)}+1)=\frac{7\hbar^2}{72 m}\,,\ \ \ \
 \ell_{(right)}(\ell_{(right)}+1)=-\frac{5\hbar^2}{72 m}\,.
 \label{sest}
 \ee
They were of an opposite sign and of different
sizes, and there were no free parameters left.
Indeed, once we return to
the
units $\hbar =2m=1$,
their CES
constraint (\ref{sest})
acquired the form
of our present Eq.~(\ref{las})
with $y= y_{(CES)}=1/6$.

For us, the latter results were encouraging.
Demonstrating that it makes sense to work with asymmetric
singular forces. By relations (\ref{sest}) we also felt inspired to a
replacement of
ansatz~(\ref{ainq})
by the more elementary one-parametric
convention~(\ref{kvIK}).
We came to the conclusion that the CES-based necessity of
working with a fixed value of $\ell=1/6$ in (\ref{sest})
was too high a price for the
technical advantage of the
user-friendliness of the model.
Thus,
we turned attention to the harmonic-oscillator
confinement~(\ref{kvIK}) because qualitatively,
the asymmetric shape
of our two-body potentials
as sampled by Figure \ref{picee3ww}
is, after all, not too different from
the shape
of the CES potential
$V_{CES}^{(IK)}(x)$ of Eq.~(\ref{VIK})
as displayed in Figure Nr.~2 of paper~\cite{Ishkha}.

In \cite{Ishkhab} the authors emphasized that
it makes good sense to search for a deeper
understanding of the connection between
the flexibility of the
spectral properties and
the absence of the tunneling
caused by the singular nature of the $X^{-2}$ barriers.
They noticed that the choice of
the
matching condition at $X = 0$
``must be adjusted according to
the specific physical context under consideration'' \cite{Ishkha}.
In our present paper the latter point has attracted due attention.
We managed to
disentangle
several related technical challenges.
In particular, we shortened the argumentation
by
omitting all of the phenomenologically motivated
references to the relativistic quantum dynamics.
After all,
even in \cite{Ishkha} we read that
the ``results show a significant difference
between the Schr\'{o}dinger and Dirac cases'', and that
``it is clear that
the reason for the discrepancy
[between the Schr\'{o}dinger and Dirac models]
is that the boundary
conditions we [=they] use for these two cases are [different, i.e.,]
due to different considerations''.
For this reason we kept the scope of our present study
restricted to
the strictly
non-relativistic systems, promising still
a very broad potential applicability \cite{scholarpedia}.

\end{document}